\documentclass[pra,aps,nopacs,onecolumn,twoside,superscriptaddress]{revtex4}

\usepackage{multirow}
\usepackage{amsmath,amsfonts,amssymb,caption,color,epsfig,graphics,graphicx,hyperref,latexsym,mathrsfs,revsymb,theorem,url,verbatim,epstopdf,enumerate}

\usepackage{amsmath,amsfonts,amssymb,caption,hyperref,color,epsfig,graphics,graphicx,latexsym,mathrsfs,revsymb,theorem,url,verbatim,epstopdf,cleveref}
\hypersetup{colorlinks,linkcolor={blue},citecolor={blue},urlcolor={red}}

\newtheorem{definition}{Definition}
\newtheorem{proposition}[definition]{Proposition}
\newtheorem{lemma}[definition]{Lemma}

\newtheorem{theorem}[definition]{Theorem}
\newtheorem{corollary}[definition]{Corollary}
\newtheorem{conjecture}[definition]{Conjecture}

\newtheorem{remark}[definition]{Remark}
\newtheorem{example}[definition]{Example}
\newtheorem{question}[definition]{Question}
\newtheorem{memo}[definition]{Memo}

\def\squareforqed{\hbox{\rlap{$\sqcap$}$\sqcup$}}
\def\qed{\ifmmode\squareforqed\else{\unskip\nobreak\hfil
\penalty50\hskip1em\null\nobreak\hfil\squareforqed
\parfillskip=0pt\finalhyphendemerits=0\endgraf}\fi}
\def\endenv{\ifmmode\;\else{\unskip\nobreak\hfil
\penalty50\hskip1em\null\nobreak\hfil\;
\parfillskip=0pt\finalhyphendemerits=0\endgraf}\fi}
\newenvironment{proof}{\noindent \textbf{{Proof.~} }}{\qed}
\def\Dbar{\leavevmode\lower.6ex\hbox to 0pt
{\hskip-.23ex\accent"16\hss}D}
\makeatletter
\def\url@leostyle{%
  \@ifundefined{selectfont}{\def\UrlFont{\sf}}{\def\UrlFont{\small\ttfamily}}}
\makeatother
\def\bcj{\begin{conjecture}}
\def\ecj{\end{conjecture}}
\def\bcr{\begin{corollary}}
\def\ecr{\end{corollary}}
\def\bd{\begin{definition}}
\def\ed{\end{definition}}
\def\bea{\begin{eqnarray}}
\def\eea{\end{eqnarray}}
\def\bem{\begin{enumerate}}
\def\eem{\end{enumerate}}
\def\bex{\begin{example}}
\def\eex{\end{example}}
\def\bim{\begin{itemize}}
\def\eim{\end{itemize}}
\def\bl{\begin{lemma}}
\def\el{\end{lemma}}
\def\bma{\begin{bmatrix}}
\def\ema{\end{bmatrix}}
\def\bpf{\begin{proof}}
\def\epf{\end{proof}}
\def\bpp{\begin{proposition}}
\def\epp{\end{proposition}}
\def\bqu{\begin{question}}
\def\equ{\end{question}}
\def\br{\begin{remark}}
\def\er{\end{remark}}
\def\bt{\begin{theorem}}
\def\et{\end{theorem}}
\def\bmm{\begin{memo}}
\def\emm{\end{memo}}

\def\btb{\begin{tabular}}
\def\etb{\end{tabular}}

\newcommand{\nc}{\newcommand}

\def\b{\beta}
\def\g{\gamma}

\def\r{\rho}

\def\ps{\psi}

\def\G{\Gamma}

 \nc{\bbA}{\mathbb{A}} \nc{\bbB}{\mathbb{B}} \nc{\bbC}{\mathbb{C}}
 \nc{\bbD}{\mathbb{D}} \nc{\bbE}{\mathbb{E}} \nc{\bbF}{\mathbb{F}}
 \nc{\bbG}{\mathbb{G}} \nc{\bbH}{\mathbb{H}} \nc{\bbI}{\mathbb{I}}
 \nc{\bbJ}{\mathbb{J}} \nc{\bbK}{\mathbb{K}} \nc{\bbL}{\mathbb{L}}
 \nc{\bbM}{\mathbb{M}} \nc{\bbN}{\mathbb{N}} \nc{\bbO}{\mathbb{O}}
 \nc{\bbP}{\mathbb{P}} \nc{\bbQ}{\mathbb{Q}} \nc{\bbR}{\mathbb{R}}
 \nc{\bbS}{\mathbb{S}} \nc{\bbT}{\mathbb{T}} \nc{\bbU}{\mathbb{U}}
 \nc{\bbV}{\mathbb{V}} \nc{\bbW}{\mathbb{W}} \nc{\bbX}{\mathbb{X}}
 \nc{\bbZ}{\mathbb{Z}}

 \nc{\bA}{{\bf A}} \nc{\bB}{{\bf B}} \nc{\bC}{{\bf C}}
 \nc{\bD}{{\bf D}} \nc{\bE}{{\bf E}} \nc{\bF}{{\bf F}}
 \nc{\bG}{{\bf G}} \nc{\bH}{{\bf H}} \nc{\bI}{{\bf I}}
 \nc{\bJ}{{\bf J}} \nc{\bK}{{\bf K}} \nc{\bL}{{\bf L}}
 \nc{\bM}{{\bf M}} \nc{\bN}{{\bf N}} \nc{\bO}{{\bf O}}
 \nc{\bP}{{\bf P}} \nc{\bQ}{{\bf Q}} \nc{\bR}{{\bf R}}
 \nc{\bS}{{\bf S}} \nc{\bT}{{\bf T}} \nc{\bU}{{\bf U}}
 \nc{\bV}{{\bf V}} \nc{\bW}{{\bf W}} \nc{\bX}{{\bf X}}
 \nc{\bZ}{{\bf Z}}

\nc{\cA}{{\cal A}} \nc{\cB}{{\cal B}} \nc{\cC}{{\cal C}}
\nc{\cD}{{\cal D}} \nc{\cE}{{\cal E}} \nc{\cF}{{\cal F}}
\nc{\cG}{{\cal G}} \nc{\cH}{{\cal H}} \nc{\cI}{{\cal I}}
\nc{\cJ}{{\cal J}} \nc{\cK}{{\cal K}} \nc{\cL}{{\cal L}}
\nc{\cM}{{\cal M}} \nc{\cN}{{\cal N}} \nc{\cO}{{\cal O}}
\nc{\cP}{{\cal P}} \nc{\cQ}{{\cal Q}} \nc{\cR}{{\cal R}}
\nc{\cS}{{\cal S}} \nc{\cT}{{\cal T}} \nc{\cU}{{\cal U}}
\nc{\cV}{{\cal V}} \nc{\cW}{{\cal W}} \nc{\cX}{{\cal X}}
\nc{\cZ}{{\cal Z}}

\nc{\hA}{{\hat{A}}} \nc{\hB}{{\hat{B}}} \nc{\hC}{{\hat{C}}}
\nc{\hD}{{\hat{D}}} \nc{\hE}{{\hat{E}}} \nc{\hF}{{\hat{F}}}
\nc{\hG}{{\hat{G}}} \nc{\hH}{{\hat{H}}} \nc{\hI}{{\hat{I}}}
\nc{\hJ}{{\hat{J}}} \nc{\hK}{{\hat{K}}} \nc{\hL}{{\hat{L}}}
\nc{\hM}{{\hat{M}}} \nc{\hN}{{\hat{N}}} \nc{\hO}{{\hat{O}}}
\nc{\hP}{{\hat{P}}} \nc{\hR}{{\hat{R}}} \nc{\hS}{{\hat{S}}}
\nc{\hT}{{\hat{T}}} \nc{\hU}{{\hat{U}}} \nc{\hV}{{\hat{V}}}
\nc{\hW}{{\hat{W}}} \nc{\hX}{{\hat{X}}} \nc{\hZ}{{\hat{Z}}}

\nc{\hn}{{\hat{n}}}

\def\lin{\mathop{\rm span}}

\def\min{\mathop{\rm min}}

\def\tr{\mathop{\rm Tr}}

\def\ox{\otimes}

\newcommand{\ket}[1]{|#1\rangle}
\newcommand{\proj}[1]{| #1\rangle\!\langle #1 |}
\newcommand{\ketbra}[2]{|#1\rangle\!\langle#2|}

\def\Dbar{\leavevmode\lower.6ex\hbox to 0pt
{\hskip-.23ex\accent"16\hss}D}

\begin{document}

\title{A complete picture of the four-party linear inequalities in terms of the $0$-entropy}

\date{\today}


\author{Zhiwei Song}\email[]{zhiweisong@buaa.edu.cn}
\affiliation{LMIB(Beihang University), Ministry of Education, and School of Mathematical Sciences, Beihang University, Beijing 100191, China}

\author{Lin Chen}\email[]{linchen@buaa.edu.cn (corresponding author)}
\affiliation{LMIB(Beihang University), Ministry of Education, and School of Mathematical Sciences, Beihang University, Beijing 100191, China}
\affiliation{International Research Institute for Multidisciplinary Science, Beihang University, Beijing 100191, China}

\author{Yize Sun}\email[]{sunyize@buaa.edu.cn(corresponding author)}
\affiliation{LMIB(Beihang University), Ministry of Education, and School of Mathematical Sciences, Beihang University, Beijing 100191, China}

\author{Mengyao Hu}\email[]{mengyaohu@buaa.edu.cn(corresponding author)}
\affiliation{LMIB(Beihang University), Ministry of Education, and School of Mathematical Sciences, Beihang University, Beijing 100191, China}
\begin{abstract}
Multipartite quantum system is complex. Characterizing the relations among the three bipartite reduced density operators $\r_{AB}$, $\r_{AC}$ and $\r_{BC}$ of a tripartite state $\r_{ABC}$ has been an open problem in quantum information. One of such relations has been reduced by [Cadney et al, LAA. 452, 153, 2014] to a conjectured inequality in terms of matrix rank, namely $r(\r_{AB}) \cdot r(\r_{AC})\ge r(\r_{BC})$ for any $\r_{ABC}$. It is denoted as open problem $41$ in the website "Open quantum problems-IQOQI Vienna". We prove the inequality, and thus establish a complete picture of the four-party linear inequalities in terms of the $0$-entropy. Our proof is based on the construction of a novel canonical form of bipartite matrices under local equivalence. We apply our result to the marginal problem and the extension of inequalities in the multipartite systems, as well as the condition when the inequality is saturated.
\end{abstract}

\maketitle


\section{Introduction}

Multipartite systems play a key role in quantum-information processing. For example, the inequality for the von Neumann entropy of reduced density operators of a multipartite has been proposed in \cite{2005A,2012Infinitely}, and the multipartite state conversion under many-copy cases has been shown under stochastic local operations and classical communications \cite{ccd2010}. Further, the relation between the distillability of entanglement of three bipartite reduced density matrices from a tripartite pure state has been studied \cite{Chen2011Multicopy, Chen2012NONDISTILLABLE,HayashiWeaker}. However, it is not easy to extend the relation to the tripartite mixed state, even we merely consider the rank of reduced density operators. It has been conjectured in \cite{chl14} that the following inequality in terms of matrix rank may hold for any tripartite mixed state $\r_{ABC}$,  
\begin{eqnarray}
\label{eq:main-inequality}
r(\r_{AB}) \cdot r(\r_{AC})\ge r(\r_{BC}),
\end{eqnarray}
where $r(M)$ denotes the rank of matrix $M$, see Figure \ref{fig:ab.bc>=ac}. The conjectured inequality in \eqref{eq:main-inequality} has been listed as the open problem in \footnote{https://oqp.iqoqi.univie.ac.at/all-rank-inequalities-for-reduced-states-of-quadripartite-quantum-states/}. It has been proven true when $r(\r_{AB})$ is at most two and three in \cite{chl14} and \cite{2020The}, respectively. In this paper we prove the inequality \eqref{eq:main-inequality} for any $\r_{ABC}$ in Theorem \ref{thm:main}. The inequality together with the inequalities constructed in \cite{chl14}, establish basic inequalities for the tradeoff among the ranks of three bipartite reduced density operators, see \eqref{eq:0entropy}-\eqref{eq:0entropy-2}. This is another point of view  in contrast to the monogamy trade-off by Bell inequalities \cite{Clauser1971Proposed,Kurzy2011Correlation, Pawlowski2009Monogamy,REPRESENTATIONS,Butterley2006Compatibility,Coffman1999Distributed,HiguchiOne}.
We thus manage to extend the results in \cite{Chen2011Multicopy, Chen2012NONDISTILLABLE, HayashiWeaker} from tripartite pure states to mixed states. Next, our results present a novel necessary condition for the marginal problem, i.e., three bipartite reduced density operators from the same tripartite state satisfy \eqref{eq:main-inequality}. Further we extend  \eqref{eq:0entropy}-\eqref{eq:0entropy-2} to multipartite systems in Lemma \ref{le:ab.bc.cd>=ad}. We also discuss the condition when the inequality \eqref{eq:main-inequality} is saturated in Lemma \ref{le:saturate}.

\begin{figure}
\centering
\includegraphics[width=0.40\textwidth]{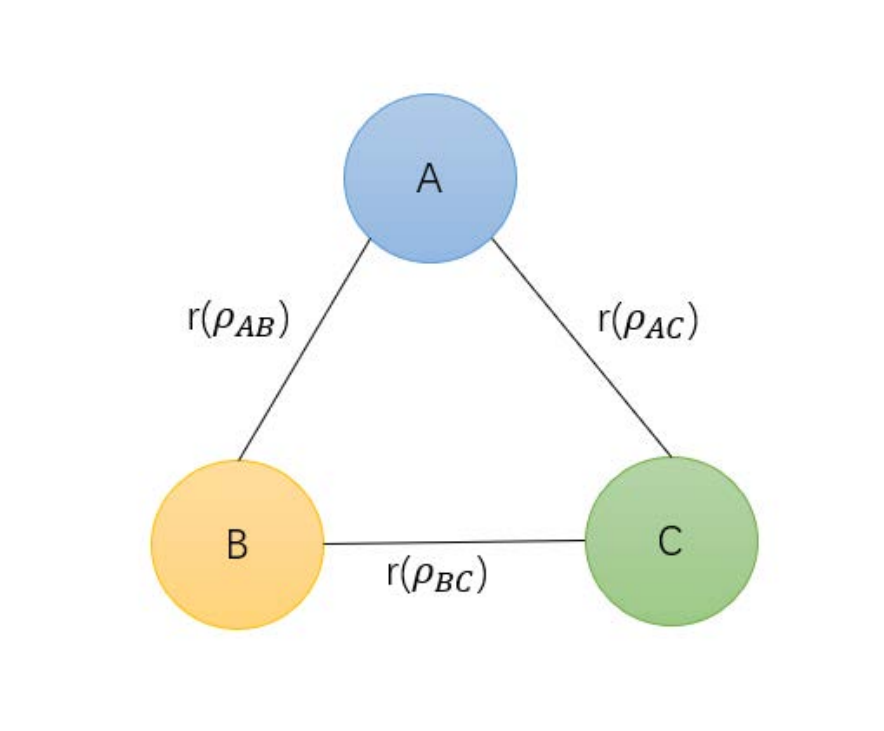}
\caption{The tripartite state $\r_{ABC}$ is conjectured to satisfy the inequality $r(\r_{AB})\cdot r(\r_{BC}) \ge r(\r_{AC})$ in terms of bipartite reduced density operators $\r_{AB}$, $\r_{BC}$, and $\r_{AC}$, and $r(M)$ denotes the rank of matrix $M$. The inequality is known to be equivalent to the $0$-entropy inequality $S_0(AB)+S_0(BC)\ge S_0(AC)$. We prove the inequality in this paper.}
\label{fig:ab.bc>=ac}
\end{figure}

We review the meaning of \eqref{eq:main-inequality} in terms of the $0$-entropy. Let $\alpha\in (0,1)\cup(1,\infty)$ and the logarithm have base two. The known $\alpha$-R\'{e}nyi entropy $S_\alpha$ of a quantum state $\rho$ is defined as
$
S_\alpha(\rho):=\frac{1}{1-\alpha}\log\tr\rho^\alpha.
$ One can verify that ${\lim_{\alpha\rightarrow1}}S_\alpha(\rho)$ is exactly the von Neumann entropy $S(\r):=-\tr (\r\log \r)$. In contrast, authors in Ref. \cite{chl14} has defined the $0$-entropy $S_0(A):=\lim_{\alpha\rightarrow0}S_\alpha(\rho_A)=\log r(\rho_A)$, the $0$-entropy vector as the $7$-dimensional vector $(r(\r_A),r(\r_B),r(\r_C),r(\r_D),
r(\r_{AB}),r(\r_{AC}),r(\r_{AD}))$ of a $4$-partite pure state of system $A,B,C,D$, as well as the subadditivity $ S_0(A)+S_0(B)\ge S_0(AB)$, $S_0(AB)+ S_0(AC)\ge S_0(A)$ as well as more inequalities, just like the counterpart inequalities of von Neumann entropy. Further, the $0$-entropy inequalities determine a cone with extremal rays characterized by eight 0-entropy vectors. In \cite{chl14}, focusing on the four-party case, authors have found six 0-entropy vectors corresponding to extremal rays, as well as the set of inequalities they correspond to.  The set turns out to be the known 0-entropy inequalities and that in \eqref{eq:main-inequality}. Hence, our proof on \eqref{eq:main-inequality} helps complete the forementioned characterization and construct a complete picture of the four-party linear inequalities for the $0$-entropy.

Technically speaking, our strategy of proving \eqref{eq:main-inequality} in Theorem \ref{thm:main} is to prove an equivalent form in terms of the partial transpose and Schmidt rank, see Conjecture \ref{cj:1}. The equivalence is proven in \cite{chl14}, and the proof of this conjecture is supported via a few basic facts from linear algebra in Lemmas \ref{lemma1} - \ref{lemma6}, as well as the construction of a novel canonical form of bipartite matrices under local equivalence in Theorem \ref{canonical}. 

The partial transpose is a positive map of extensive applications in quantum information.
Firstly it is known that a separable state is positive partial transpose (PPT), and it is the most efficient method of detecting entanglement so far \cite{hhh96}. Next, the two-qutrit PPT entangled states were constructed in 1997 \cite{horodecki1997}, and such states of rank four have been characterized \cite{Chen2012Equivalence,2011Three}. The PPT entanglement represents quantum resources which cannot be distillable into pure entangled states under local operations and classical communications (LOCC). What's more, 
bipartite non-PPT states of rank at most four turn out to be distillable \cite{Chen2008Rank,Chen2012Distillability,Lin2016Non}, and some non-PPT states are conjectured to be non distillable
\cite{hh1999, Divincenzo2000Evidence,5508622,QIAN2021139, 2020Five}. On the other hand, the Schmidt rank is a basic parameter of characterizing bipartite pure states, and has been extended to multipartite pure states as an entanglement monotone \cite{Eisert2000The,ccd2010}. Our result of proving \eqref{eq:main-inequality} shows novel understanding of the forementioned quantum-information applications in terms of partial transpose and Schmidt rank.

The rest of this paper is organized as follows. In Sec. \ref{sec:pre} we introduce the preliminary facts in Lemmas \ref{lemma1} - \ref{lemma6} for the proof of inequality \eqref{eq:main-inequality}. We show the proof in Theorem \ref{thm:main} supported by Theorem \ref{canonical} in Sec. \ref{sec:proof}. Then we apply our result in Lemmas \ref{le:ab.bc.cd>=ad} and \ref{le:saturate} in Sec. \ref{sec:app}. Finally we conclude in Sec. \ref{sec:con}.

\section{Preliminaries}
\label{sec:pre}

In this section we introduce the preliminary knowledge and facts of this paper. 
Let $\bbM_{m,n}$ be the set of $m\times n$ complex matrices, and $\bbM_n:=\bbM_{n,n}$. Let $I_{n}$ be the order-$n$ identity matrix.
We denote $M^T$ as the transpose of matrix $M$. We denote the bipartite operator $M\in \bbM_{m_1,n_1}\ox\bbM_{m_2,n_2}$ as $M=\sum^{m_1}_{i=1}\sum^{n_1}_{j=1}\ketbra{i}{j}\otimes M_{i,j}$ with $M_{i,j}\in \bbM_{m_2,n_2}.$ We denote the partial transpose of $M$ w.r.t. system $A$ and $B$ as $M^{\Gamma_A}=\sum^{m_1}_{i=1}\sum^{n_1}_{j=1}\ketbra{j}{i}\otimes M_{i,j}$, and $M^{\Gamma_B}=\sum^{m_1}_{i=1}\sum^{n_1}_{j=1}\ketbra{i}{j}\otimes M_{i,j}^T$, respectively. Next, the number of linearly independent blocks $M_{i,j}$'s is referred to the Schmidt rank $Sr:=Sr(M)$ of $M$. One can derive that
$Sr(M)
=
Sr(M^{\G_A})=Sr(M^{\G_B})
=Sr(M^T).	
$ 
It has been proven by \cite{chl14} that the conjectured inequality in \eqref{eq:main-inequality} is equivalent to the following conjecture.

\bcj
\label{cj:1}
Let $M\in \bbM_{m_1,n_1}\ox\bbM_{m_2,n_2}$. Then
\bea
\label{eq:cj2}
r(M^{\G_B})\le Sr(M)\cdot r(M).
\eea
The conjecture is equivalent to $r(M^{\G_A})\le Sr(M)\cdot r(M)$. If we write $M=\sum^K_{i=1} R_i \ox S_i^T$, where $R_1,\dots,R_K\in \bbM_{m_1,n_1}$ are linearly independent, and
$S_1,\dots,S_K\in \bbM_{m_2,n_2}$ are linearly independent, then the inequality is also equivalent to
$ r \bigg( \sum^K_{i=1} R_i \ox S_i \bigg)
\le
K\cdot r \bigg( \sum^K_{i=1} R_i \ox S_i^T \bigg)$.
\qed
\ecj

Next it follows from Theorem 5 of \cite{chl14} that Conjecture \ref{cj:1} holds for $Sr(M)\le2$. It has been proved in Theorem 2 of \cite{2020The} that Conjecture \ref{cj:1} holds for $Sr(M)=3$. Nevertheless, these proven cases do not contribute to our proof in the next section. To explain our proof, we present the following definition.

\begin{definition}
	\label{def:equivalent}
	We denote locally equivalent $M,N$ as $M\sim N$ if there exist invertible product matrices $U\otimes V$ and $W\otimes X$ such that $(U\otimes V)M(W\otimes X)=N$. 
	\qed
\end{definition}

So $M$ and $N$ have the same rank and Schmidt rank, one can also show that $M^\G$ and $N^\G$ have the same rank and Schmidt rank. Hence proving $M$ satisfies Conjecture \ref{cj:1} is equivalent to prove $N$ satisfies Conjecture \ref{cj:1}, and we shall frequently use this fact in the next section. In the rest of this section, we present five preliminary lemmas used for the proof of next section.  The following two lemmas can be straightforwardly proven using the basic matrix theory.  
\begin{lemma}
	\label{lemma1}	
	The following inequalities hold for any  block matrix
	\begin{eqnarray}
	\label{baseineq1}
	r(A_1)\le r(\bma A_1&A_2&\cdots&A_n
	\ema)\le r(A_1)+\cdots+r(A_n),
	\end{eqnarray}
	\begin{eqnarray}
	\label{baseineq2}
	r(A)+r(C)\le r(\bma A&0\\
	B&C\ema).
	\end{eqnarray}
	\qed
\end{lemma}
\begin{lemma}
	\label{lemmaadd}
	(i) Suppose $A_1,\cdots,A_n\in \bbM_{m_2,n_2}$ are linearly independent, $R\in \bbM_{n_2}$ is invertible. Then $A_1\cdot R,\cdots,A_n\cdot R$ are linearly independent.
	
	(ii) Suppose $A_1,\cdots,A_n,A_{n+1}\in \bbM_{m_2,n_2}$, and $R\in \bbM_{n_2}$ is invertible. If $A_{n+1}\in \lin \{A_1,\cdots,A_n\}$, then $A_{n+1}\cdot R\in \lin \{A_1\cdot R,\cdots,A_n\cdot R\}$.
	\qed
\end{lemma}
The following lemma
is assertion (a) of Theorem 6 in \cite{2020The}. 
\begin{lemma}
	\label{lemma10}
	Suppose one block-row or block-column of the Schmidt-rank-K block matrix $M$ has $K$ linearly independent blocks. Then $M$ satisfies Conjecture \ref{cj:1}. 
	\qed
\end{lemma}

We next present another lemma.
\begin{lemma}
\label{lemma11}
(i)	Suppose 
	\begin{eqnarray}
	\notag
	M=\bma M_{1,1}&M_{1,2}&M_{1,3}&\cdots &M_{1,n_{1}}\\
	\vdots &\vdots &\vdots &\ddots &\vdots\\M_{t,1}&M_{t,2}&M_{t,3}&\cdots&M_{t,n_{1}}\\0&M_{t+1,2}&M_{t+1,3}&\cdots&M_{t+1,n_1}\\\vdots&\vdots&\vdots&\ddots&\vdots\\
	0&M_{m_{1},2}&M_{m_{1},3}&\cdots &M_{m_{1},n_{1}} \ema,
	\end{eqnarray}
	where $M_{1,1},\cdots,M_{t,1}$ are linearly independent and there exists a block $M_{i,j}$ that is linearly independent with the $t$ blocks for $t+1\le i\le m_1,2\le j\le n_1$. Then by multiplying an appropriate constant to the $j$-th block-column of $M$ and adding to the first block-column, the new first block-column has at least $t+1$ linearly independent blocks.

(ii) Suppose 
\begin{eqnarray}
\notag
M=\bma M_{1,1}&M_{1,2}&M_{1,3}&\cdots &M_{1,n_{1}}\\
\vdots &\vdots &\vdots &\ddots &\vdots\\M_{t,1}&M_{t,2}&M_{t,3}&\cdots&M_{t,n_{1}}\\M_{t+1,1}&M_{t+1,2}&M_{t+1,3}&\cdots&M_{t+1,n_1}\\\vdots&\vdots&\vdots&\ddots&\vdots\\
M_{s,1}&M_{s,2} & M_{s,3}& \cdots& M_{s,n_1}
\ema,
\end{eqnarray}
where $s>t$, and the $s+t$ matrices $M_{1,1},\cdots,M_{s,1},M_{1,2},\cdots, M_{t,2}$ are linearly independent. Any $M_{i,2}$ is spanned by $M_{1,1},\cdots,M_{s,1}$, where $t+1\le i\le s$. At the same time, there exists a block $M_{i,j}$ that is linearly independent with the $s+t$ matrices for $t+1\le i\le s,3\le j\le n_1$. Then by multiplying an appropriate constant to the $j$-th block column of $M$ and adding to the second block-column, the new second block-column has at least $t+1$ linearly independent blocks and they are linearly independent with $M_{1,1},\cdots,M_{s,1}$.

(iii) Suppose
\begin{eqnarray}
\notag
M=\bma M_{1,1}&M_{1,2}&\cdots&M_{1,n}&\cdots &M_{1,n_{1}}\\
\vdots &\vdots &\ddots &\vdots&\ddots &\vdots \\M_{t_n,1}&M_{t_n,2}&\cdots&M_{t_n,n}&\cdots&M_{t_n,n_{1}}\\M_{t_n+1,1}&M_{t_n+1,2}&\cdots&M_{t_n+1,n}&\cdots&M_{t_n+1,n_1}\\\vdots&\vdots&\ddots&\vdots&\ddots&\vdots\\
M_{t_2,1}&M_{t_2,2} &\cdots& M_{t_2,n}& \cdots& M_{t_2,n_1}\\
\vdots&\vdots&\ddots&\vdots&\ddots&\vdots\\
M_{t_1,1}&M_{t_1,2} &\cdots& M_{t_1,n}& \cdots& M_{t_1,n_1}
\ema,
\end{eqnarray}
where $t_1\ge t_2\ge \cdots\ge t_{n-1}>t_n$, and the $t_1+t_2+\cdots+t_n$ matrices $M_{1,1},\cdots,M_{t_1,1}$, $M_{1,2},\cdots, M_{t_2,2}, \cdots$, $M_{1,n},\cdots,M_{t_n,n}$ are linearly independent. Any $M_{i,n}$ is spanned by $M_{1,1},\cdots,M_{t_1,1}$, $M_{1,2},\cdots, M_{t_2,2}, \cdots$, $M_{1,n-1},\cdots,M_{t_{n-1},n-1}$, where $t_n+1\le i\le t_1$. At the same time, there exists a block $M_{i,j}$ that is linearly independent with the $t_1+t_2+\cdots+t_n$ matrices for $t_n+1\le i\le t_1,n+1\le j\le n_1$. Then by multiplying an appropriate constant to the $j$-th block column of $M$ and adding to the $n$-th block-column of $M$, the new $n$-th block-column has at least $t_n+1$ linearly independent blocks and they are linearly independent with $M_{1,1},\cdots,M_{t_1,1},M_{1,2},\cdots,M_{t_2,2}$, $\cdots,M_{1,n-1},\cdots,M_{t_{n-1},n-1}$.
\end{lemma}
\begin{proof}
(i) Without loss of generality, we assume that $M_{t+1,2}$ is linearly independent with $M_{1,1},\cdots,M_{t,1}$. We have two cases, namely (\hyperlink{A}{A}) and (\hyperlink{B}{B}), as follows.

\hypertarget{A}{(A)}. Suppose any $M_{i,2}$ is spanned by $M_{1,1},\cdots,M_{t,1},M_{t+1,2}$, where $1\le i\le t$. We first apply block-row operations on $M$ and still use $M_{i,j}$ to denote the new blocks, such that $M_{1,1},\cdots,M_{t,1}$ are  linearly independent, and any $M_{i,2}$, where $1\le i\le t$, is the combinations of the $t$ matrices. Thus we set
\begin{eqnarray}
\label{mi2}
M_{i,2}=a_{1i}M_{1,1}+\cdots+a_{ti}M_{t,1},
\end{eqnarray}
where $a_{1i},\cdots,a_{ti}$ are coefficients, $1\le i\le t$.
Next, we multiply a nonzero constant $k$ to the second block-column of $M$ and add to the first block-column.  Denote $M_{1,1}',\cdots,M_{t,1}',M_{t+1,1}'$ as the new blocks in the first block-column, we have
\begin{eqnarray}
\label{m'}
M_{i,1}'=M_{i,1}+kM_{i,2}
\end{eqnarray}
holds for $1\le i\le t+1$. 
Assume that
\begin{eqnarray}
\label{coe1}
&&\eta_1M_{1,1}'+\eta_2M_{2,1}'\cdots+\eta_tM_{t,1}'+\eta_{t+1}M_{t+1,1}'=0,
 \end{eqnarray}
where $\eta_1,\eta_2,\cdots,\eta_t,\eta_{t+1}$ are complex numbers. Using (\ref{mi2}) and (\ref{m'}), we have
\begin{eqnarray}
\label{co}
\notag
&&(\eta_1(ka_{11}+1)+\eta_2ka_{12}+\cdots+\eta_tka_{1t})M_{1,1}\\
\notag
+&&(\eta_1ka_{21}+\eta_2(ka_{22}+1)+\cdots+\eta_tka_{2t})M_{2,1}\\
\notag
+&& \cdots\\
\notag
+&&(\eta_1ka_{t1}+\eta_2ka_{22}+\cdots+\eta_t(ka_{tt}+1))M_{t,1}\\
+&&\eta_{t+1}kM_{t+1,2}=0.
\end{eqnarray}
Since $M_{1,1},\cdots,M_{t,1},M_{t+1,2}$ are linearly independent, from (\ref{co}) we have 
\begin{eqnarray}
\label{coe}
\bma \begin{array}{ccccc} ka_{11}+1&ka_{12}&\cdots&ka_{1t}&0\\
ka_{21}&ka_{22}+1&\cdots&ka_{2t}&0\\
\vdots&\vdots&\ddots&\vdots&\vdots\\
ka_{t1}&ka_{t2}&\cdots&ka_{tt}+1&0\\
0&0&\cdots&0&k
\end{array}
\ema\cdot
\bma \eta_1\\\eta_2\\ \vdots\\\eta_t\\\eta_{t+1} \ema
=\bma 0\\0\\\vdots\\0\\0 \ema.
\end{eqnarray}
For the linear equations (\ref{coe}), we can find an appropriate $k'$ such that 
\begin{eqnarray}
\det 
\left| \begin{array}{ccccc} k'a_{11}+1&k'a_{12}&\cdots&k'a_{1t}&0\\
k'a_{21}&k'a_{22}+1&\cdots&k'a_{2t}&0\\
\vdots&\vdots&\ddots&\vdots&\vdots\\
k'a_{t1}&k'a_{t2}&\cdots&k'a_{tt}+1&0\\
0&0&\cdots&0&k'
\end{array}
\right| \neq 0,
\end{eqnarray}
and hence $\eta_1=\eta_2=\cdots=\eta_t=\eta_{t+1}=0$.  Let $k=k'$, then $M_{1,1}',\cdots,M_{t,1}',M_{t+1,1}'$ are linearly independent in terms of (\ref{coe1}). We have finished the proof.

\hypertarget{B}{(B)}. Suppose there exists at least one $M_{i,2}$ that is linearly independent with $M_{1,1},\cdots,M_{t,1},M_{t+1,2}$, where $1\le i\le t$. Without loss of generality, assume that $M_{1,1},\cdots,M_{t,1},M_{t+1,2},M_{1,2},\cdots,M_{v,2}$ are linearly independent, where $1\le v\le t$. At the same time, any $M_{i,2}$ is spanned by the $t+1+v$ matrices for $v+1\le i\le t$. 
Next, apply block-row operations on $M$ such that
$M_{i,2}$, where $v+1\le i\le t$, is the combination of $M_{1,1},\cdots,M_{t,1}$, and the coefficients are still denoted by those in (\ref{mi2}). 
By multiplying $k'$ in case (\hyperlink{A}{A}) to the second block-column of $M$ and adding to the first block-column, we obtain that $M_{v+1,1}',\cdots,M_{t+1,1}'$ are linearly independent. Since $k'$ is nonzero, we also have $M_{1,1}',\cdots,M_{v,1}'$ are linearly independent, and they are linearly independent with $M_{v+1,1}',\cdots,M_{t,1}'$. We have finished the proof.\\

	(ii) Without loss of generality, assume that $M_{t+1,3}$ is linearly independent with the $s+t$ matrices. We have two cases, namely (\hyperlink{C}{C}) and (\hyperlink{D}{D}), as follows.
	
\hypertarget{C}{(C)}. Suppose any $M_{i,3}$ is spanned by $M_{1,1},\cdots,M_{s,1},M_{1,2},\cdots, M_{t,2},M_{t+1,3}$,  where $1\le i\le t$. We first apply block-row operations on $M$, such that any $M_{i,3}$ is spanned by $M_{1,1},\cdots,M_{s,1},M_{1,2},\cdots, M_{t,2}$, where $1\le i\le t$. At the same time, $M_{t+1,2}$ is still spanned by $M_{1,1},\cdots,M_{s,1}$. Thus we set
 	\begin{eqnarray}
 	\label{mi3}
	M_{i,3}=b_{1i}M_{1,1}+\cdots+b_{si}M_{s,1}+c_{1i}M_{1,2}+\cdots+c_{ti}M_{t,2},
	\end{eqnarray}
	where $b_{1i},\cdots,b_{si},c_{1i},\cdots,c_{ti}$ are coefficients, $1\le i\le t$.
	Set
	\begin{eqnarray} 
	\label{spec}
	M_{t+1,2}=d_{1}M_{1,1}+\cdots+d_{s}M_{s,1},
	\end{eqnarray}
	where $d_{1},\cdots,d_{s}$ are coefficients. Next, we multiply a nonzero constant $k$ to the third block-column of $M$ and add to the second block-column.  Denote $M_{1,2}'',\cdots,M_{t,2}'',M_{t+1,2}''$ as the new blocks in the second block-column, we have
	\begin{eqnarray}
	\label{m''}
	M_{i,2}''=M_{i,2}+kM_{i,3}
	\end{eqnarray}
holds for $1\le i\le t+1$. Assume that 
	\begin{eqnarray}
	\label{coe2}
\zeta_1M_{1,1}+\cdots+\zeta_sM_{s,1}+\gamma_{1}M_{1,2}''+\cdots+\gamma_tM_{t,2}''+\gamma_{t+1}M_{t+1,2}''=0,
	\end{eqnarray}
where $\zeta_1,\cdots,\zeta_s,\gamma_{1},\cdots,\gamma_t,\gamma_{t+1}$ are complex numbers. Using (\ref{mi3}), (\ref{spec}) and (\ref{m''}), we have
	\begin{eqnarray}
	\label{coe3}
	\notag
	&&(\zeta_1+\gamma_1kb_{11}+\cdots+\gamma_tkb_{1t}+\gamma_{t+1}d_1)M_{1,1}\\
	\notag
	+&&\cdots \\
	\notag
	+&&(\zeta_s+\gamma_1kb_{s1}+\cdots+\gamma_tkb_{st}+\gamma_{t+1}d_s)M_{s,1}\\
	\notag 
	+&&(\gamma_1(kc_{11}+1)+\cdots+\gamma_tkc_{1t})M_{1,2}\\
	\notag
	+&& \cdots\\
	\notag
	+&&(\gamma_1kc_{t1}+\cdots+\gamma_t(kc_{tt}+1))M_{t,2}\\
	+&&k\eta_{t+1}M_{t+1,3}=0.
	\end{eqnarray}
Since $M_{1,1},\cdots,M_{s,1}$, $M_{1,2},\cdots, M_{t,2},M_{t+1,3}$ are linearly independent, from (\ref{coe3}) we have
	\begin{eqnarray}
	\label{eq:88}
\bma
\begin{array}{ccccccccc}
	\multicolumn{3}{c}{\multirow{3}{*}{\Large ${I_s}$}}&kb_{11}&\cdots&kb_{1t}&d_1\\
	 & & &\vdots&\ddots&\vdots&\vdots\\
	 & & &kb_{s1}&\cdots&kb_{st}&d_s\\
	0&\cdots&0&kc_{11}+1&\cdots&kc_{1t}&0\\
	\vdots&\ddots&\vdots&\vdots&\ddots&\vdots&\vdots\\
	0&\cdots&0&kc_{t1}&\cdots&kc_{tt}+1&0\\
	0&\cdots&0&0&\cdots&0&k
	\end{array}
\ema\cdot
\bma \zeta_1\\\vdots\\ \zeta_s\\ \gamma_1\\ \vdots\\ \gamma_{t}\\\gamma_{t+1} \ema=\bma 0\\\vdots\\0\\0\\ \vdots\\0\\0 \ema	.
	\end{eqnarray}
	For the linear equations (\ref{eq:88}), we can find an appropriate $k''$ such that 
	\begin{eqnarray}
\det 
	\left| \begin{array}{ccccc} k''c_{11}+1&k''c_{12}&\cdots&k''c_{1t}&0\\
	k''c_{21}&k''c_{22}+1&\cdots&k''c_{2t}&0\\
	\vdots&\vdots&\ddots&\vdots&\vdots\\
	k''c_{t1}&k''c_{t2}&\cdots&k''c_{tt}+1&0\\
	0&0&\cdots&0&k''
	\end{array}
	\right| \neq 0,
	\end{eqnarray}
and hence $\zeta_1=\cdots=\zeta_s=\gamma_1=\cdots=\gamma_t=\gamma_{t+1}=0$.
	Let $k=k''$, then
	$M_{1,1},\cdots,M_{s,1},M_{1,2}'',\cdots, M_{t,2}'',M_{t+1,2}''$ are linearly independent in terms of (\ref{coe2}).
	We have finished the proof.

\hypertarget{D}{(D)}. Suppose there exists at least one $M_{i,3}$ that is linearly independent with $M_{1,1},\cdots,M_{s,1},M_{1,2},\cdots,M_{t,1},M_{t+1,3}$, where $1\le i\le t$. Without loss of generality, assume that $M_{1,1},\cdots,M_{s,1},M_{1,2},\cdots,M_{t,2},M_{t+1,3},M_{1,3},\cdots,M_{v,3}$ are linearly independent, where $1\le v\le t$. At the same time, any $M_{i,3}$ is spanned by the $s+t+1+v$ matrices for $v+1\le i\le t$. 
Next, apply block-row operations on $M$ such that
$M_{i,3}$, where $v+1\le i\le t$, is the combination of $M_{1,1},\cdots,M_{s,1},M_{1,2},\cdots,M_{t,2}$, and the coefficients are still denoted by those in (\ref{mi3}). 
By multiplying $k''$ in case (\hyperlink{C}{C}) to the third block-column of $M$ and adding to the second block-column, we obtain that $M_{v+1,2}'',\cdots,M_{t+1,2}''$ are linearly independent and they are linearly independent with $M_{1,1},\cdots,M_{s,1}$. Since $k''$ is nonzero, we also have $M_{1,2}'',\cdots,M_{v,2}''$ are linearly independent and they are linearly independent with $M_{v+1,2}'',\cdots,M_{t+1,2}''$ and the $s$ blocks in the first block-column. We have finished the proof.\\

(iii) The proof is similar to that of (ii).
\end{proof}

Finally we present another  observation.
\begin{lemma}
	\label{lemma6}
	Suppose the $m\times 1$ block matrix $\bma P_1\\ P_2\\ \vdots\\ P_m \ema$ has $n$ columns and $k$ linearly independent columns, $k\le n$. Then for any matrix $Q\in \lin \{ P_1,\cdots, P_m\}$, Q has $k$ linearly independent columns at most. Through elementary column operations on $Q$, $Q$ has $k$ nonzero columns at most.
\end{lemma}
\begin{proof}
	Without loss of generality, assume that the first $k$ columns of $\bma P_1\\ P_2\\ \vdots\\ P_m \ema$ are linearly independent and the remaining columns are in the span of them. Define $\beta_j$ as the $j$-th column of $\bma P_1\\ P_2\\ \vdots\\ P_m \ema$, where $1\le j\le n$. We set
	\begin{eqnarray}
	\label {expr1}
	\notag \beta_{k+1}&=&c_{1,k+1}\beta_1+c_{2,k+1}\beta_2+\cdots+c_{k,k+1}\beta_k,\\
	\notag \beta_{k+2}&=&c_{1,k+2}\beta_1+c_{2,k+2}\beta_2+\cdots+c_{k,k+2}\beta_k,\\
	\notag
	&\vdots& \\
	\beta_{n}&=&c_{1,n}\beta_1+c_{2,n}\beta_2+\cdots+c_{k,n}\beta_k,
	\end{eqnarray}
	where $c_{a,b}$ is the combination cofficient, $1\le a\le k$,  and $k+1\le b\le n$. 
	Define $\beta_{j,i}$ as the $j$-th column of $P_i$, where $1\le j\le n$.
	From (\ref{expr1}) we obtain that
	\begin{eqnarray}
	\label {expr33}
	\notag \beta_{k+1,i}&=&c_{1,k+1}\beta_{1,i}+c_{2,k+1}\beta_{2,i}+\cdots+c_{k,k+1}\beta_{k,i},\\
	\notag \beta_{k+2,i}&=&c_{1,k+2}\beta_{1,i}+c_{2,k+2}\beta_{2,i}+\cdots+c_{k,k+2}\beta_{k,i},\\
	\notag
	&\vdots& \\
	\beta_{n,i}&=&c_{1,n}\beta_{1,i}+c_{2,n}\beta_{2,i}+\cdots+c_{k,n}\beta_{k,i}
	\end{eqnarray}
	hold for any $1 \le i\le m$.
	Since $Q$ is in the span of $P_1,\cdots, P_m$, assume 
	\begin{eqnarray}
	\label{Q}
	Q=p_1P_1+\cdots +p_mP_m,
	\end{eqnarray}
	with the combination cofficients $p_1,\cdots, p_m$. Define $\alpha_j$ as the $j$-th column of $Q$, where $1\le j\le n$. From (\ref{Q})  we have
	\begin{eqnarray}
	\label{ak+1}
	\notag
	\alpha_{k+1}&&=p_1\beta_{k+1,1}+\cdots +p_m\beta_{k+1,m} \\
	\notag
	&&=p_1(c_{1,k+1}\beta_{1,1}+c_{2,k+1}\beta_{2,1}+\cdots+c_{k,k+1}\beta_{k,1})\\
	\notag
	&&+\cdots\\
	&&+p_m(c_{1,k+1}\beta_{1,m}+c_{2,k+1}\beta_{2,m}+\cdots+c_{k,k+1}\beta_{k,m}).
	\end{eqnarray}
	From (\ref{expr33}) and (\ref{ak+1}), we have
	\begin{eqnarray}
	\label{ak+2}
	\notag
	\alpha_{k+1}&&=c_{1,k+1}(p_1\beta_{1,1}+\cdots+p_m\beta_{1,m})\\
	\notag
	&&+\cdots\\
	\notag
	&&+c_{k,k+1}(p_1\beta_{k,1}+\cdots+p_m\beta_{k,m})\\
	&&=c_{1,k+1}\alpha_1+\cdots+c_{k,k+1}\alpha_k.
	\end{eqnarray}
	We can use the same way of obtaining (\ref{ak+2}) to prove that
	\begin{eqnarray}
	\notag
	\alpha_{k+2}&=&c_{1,k+2}\alpha_1+\cdots+c_{k,k+2}\alpha_k,\\
	\notag
	&\vdots &\\
	\alpha_{n}&=&c_{1,n}\alpha_1+\cdots+c_{k,n}\alpha_k.
	\end{eqnarray}
	Hence any $\alpha_j$ is spanned by $\alpha_1,\cdots,\alpha_k$, where $k+1\le j\le n$. By elementary column operations on $Q$, we obtain that $\alpha_j$ becomes zero column for any $k+1\le j\le n$. And it has $k$ nonzero columns at most. We finished the proof.
\end{proof}

Using preceding lemmas, we present the proof of  Conjecture \ref{cj:1} in the next section. This is equivalent to proving \eqref{eq:main-inequality} in terms of \cite{chl14}.

\section{Proof} 
\label{sec:proof}

For any block matrix $M\in \bbM_{m_1,n_1}\ox\bbM_{m_2,n_2}$ with Schmidt rank $K \le m_1\cdot n_1$,  we write 
\begin{eqnarray}
\label{M}
M=\bma M_{1,1}&M_{1,2}&M_{1,3}&\cdots &M_{1,n_{1}}\\ M_{2,1}&M_{2,2}&M_{2,3}&\cdots &M_{2,n_{1}}\\
\vdots &\vdots &\vdots &\ddots &\vdots\\
M_{m_{1},1}&M_{m_{1},2}&M_{m_{1},3}&\cdots &M_{m_{1},n_{1}} \ema,
\end{eqnarray}
where $M$ has $K$ linearly independent blocks.

Our proof of Conjecture \ref{cj:1} is divided into three subsections. In subsection \ref{subsec: part1}, we proof Theorem \ref{canonical}, transforming $M$ to a canonical form $N$ in (\ref{N}) up to local equivalence. In subsection \ref{subsec: part2}, we present an equivalent form of $N$ in (\ref{Np}). In subsection \ref{subsec: part3}, we prove the conjecture by induction.
\subsection{A canonical form of $M$}
\label{subsec: part1}
We first present a set $\bbM_{canonical} \subseteq \bbM_{m_1,n_1}\ox\bbM_{m_2,n_2}$. Any block matrix $N\in \bbM_{canonical}$ can be written as 
\begin{eqnarray}
\label{N}
N=\bma 
\begin{array}{c|c|ccc|cc}
N_{1,1}&N_{1,2}&\multicolumn{1}{c|}{N_{1,3}}&\multicolumn{1}{c|}{\cdots}&N_{1,p}&\multicolumn{2}{c}{\multirow{5}{*}{$A_p$}}\\
\vdots &\vdots &\multicolumn{1}{c|}{\vdots} &\multicolumn{1}{c|}{\ddots}&\vdots \\N_{{k_p},1}&N_{{k_p},2}&\multicolumn{1}{c|}{N_{{k_p},3}}&\multicolumn{1}{c|}{\cdots}&N_{{k_p},p} \\
\cline{5-5}
\vdots &\vdots &\multicolumn{1}{c|}{\vdots} &\multicolumn{1}{c|}{\vdots} & \multicolumn{1}{c|}{\multirow{2}{*}{$A_{p-1}$}}\\
\cline{4-4}
\cline{6-7}

\vdots&\vdots&\multicolumn{1}{c|}{\vdots}& \multicolumn{1}{c|}{}&\multicolumn{1}{c}{} \\
\cline{5-7}
N_{{k_3},1} &N_{{k_3},2} &\multicolumn{1}{c|}{N_{{k_3},3}}& \multicolumn{4}{c}{\multirow{3}{*}{$\cdots$}}\\
\cline{3-3}
\vdots&\vdots&\multicolumn{1}{c|}{}\\
N_{{k_2},1}&N_{{k_2},2}&\multicolumn{1}{c|}{} & \\
\cline{2-2}
\cline{4-7}
\vdots &  &\multicolumn{5}{c}{\multirow{2}{*}{$A_2$}}\\N_{k_{1},1}&  \\
\cline{1-1}
\cline{3-7}
0&\multicolumn{4}{c}{\multirow{3}{*}{$A_1$}}\\
\vdots \\
0
\end{array}
\ema,
\end{eqnarray}
where $N_{1,1},\cdots,N_{k_1,1}, N_{1,2}, \cdots, N_{k_2,2},\cdots,N_{1,p},\cdots, N_{k_p,p}$ are linearly independent, and other blocks are in the span of them, $1\le p\le n_1$, and
\begin{eqnarray}
\label{sum}
1\le k_p\le \cdots \le k_2 \le k_1\le m_1, 
\end{eqnarray}
\begin{eqnarray}
\label{K}
k_1+k_2+\cdots+k_p=K. 
\end{eqnarray}
At the same time, every block in $A_i$ is spanned by $N_{1,1},\cdots,N_{k_1,1}$, $N_{1,2},\cdots,N_{k_2,2},\cdots,N_{1,i},\cdots, N_{k_i,i}$, namely
\begin{eqnarray}
\label{Aspan}
A_i\tilde{\in} \lin \{N_{1,1},\cdots,N_{k_1,1},N_{1,2},\cdots,N_{k_2,2},N_{1,i},\cdots, N_{k_i,i}\}
\end{eqnarray}
holds for any $1\le i\le p$.

Note that if $m_1=k_1$, then the zero blocks below $N_{k_1,1}$ disappear and $A_1$ becomes a $(k_1-k_2)\times 1$ rectangular block matrix. If $k_1=k_2$,  then $A_1$ becomes a $(m_1-k_1)\times (n_1-1)$ rectangular block matrix and $A_2$ becomes a $(k_2-k_3)\times 1$ rectangular block matrix. If there exists $i$ ($2\le i<p-1$) such that $k_i=k_{i+1}$, then $A_i$ becomes a $(k_{i-1}-k_i)\times (n_1-i)$ rectangular block matrix, and $A_{i+1}$ becomes a $(k_{i+1}-k_{i+2})\times 1$ rectangular block matrix. If $k_{p-1}=k_p$, then $A_{p-1}$ becomes a $(k_{p-2}-k_{p-1})\times (n_1-p+1)$ rectangular block matrix, and $A_{p}$ becomes a $k_{p}\times (n_1-p)$ rectangular block matrix.

We present the main result of this subsection.
\begin{theorem}
	\label{canonical}
	For any bipartite matrix $M$ in (\ref{M}), there exists a block matrix $N\in \bbM_{canonical}$ in (\ref{N}), such that $M$ is locally equivalent to $N$.
\end{theorem}

\begin{proof}
	Our first aim is to obtain that
	\begin{eqnarray}
	\label{M'}
	M\sim M'=
	\bma 
	\begin{array}{c| cccc}
	M_{1,1}&M_{1,2}&M_{1,3}&\cdots &M_{1,n_{1}}\\ M_{2,1}&M_{2,2}&M_{2,3}&\cdots &M_{2,n_{1}}\\
	\vdots &\vdots &\vdots &\ddots &\vdots\\
	M_{k_{1},1}&M_{k_{1},2}&M_{k_{1},3}&\cdots&M_{k_{1},n_{1}}\\
	\hline
	0&
	\multicolumn{4}{|c}{\multirow{3}{*}{$A_1$}}\\
	\vdots& \\
	0
	\end{array}
	\ema,
	\end{eqnarray}
	where $k_1\le m_1$, and $M_{1,1},\cdots, M_{k_1,1}$ are linearly independent, and all the blocks in $A_1$ are spanned by $M_{1,1},\cdots, M_{k_1,1}$, namely
	\begin{eqnarray}
	\label{A1span}
	A_1\tilde{\in}\lin \{M_{1,1},\cdots,M_{k_1,1}\}.
	\end{eqnarray}
	Note that the blocks $M_{i,j}$ in (\ref{M'}) may be different from those in (\ref{M}), and we use the same symbols $M_{i,j}$, when there is no confusion, throughtout the proof. We apply the following three steps, namely {\textbf{Steps}} \hyperlink{step1}{1}-\hyperlink{step3}{3}, to achieve this aim.
	
	{\textbf{Step}} \hypertarget{step1}{{\textbf{1}}}
	Consider the first block-column of $M$ in (\ref{M}). Assume that it has $s_1$ linearly independent blocks, $s_1\le m_1$.
	Then by block-row operations on $M$, we obtain that 
	\begin{eqnarray}
	\label{M2}
	M\sim M_1=\bma M_{1,1}&M_{1,2}&M_{1,3}&\cdots &M_{1,n_{1}}\\ M_{2,1}&M_{2,2}&M_{2,3}&\cdots &M_{2,n_{1}}\\
	\vdots &\vdots &\vdots &\ddots &\vdots\\M_{s_{1},1}&M_{s_{1},2}&M_{s_{1},3}&\cdots&M_{s_{1},n_{1}}\\0&M_{s_{1}+1,2}&M_{s_{1}+1,3}&\cdots&M_{s_{1}+1,n_1}\\\vdots&\vdots&\vdots&\ddots&\vdots\\
	0&M_{m_{1},2}&M_{m_{1},3}&\cdots &M_{m_{1},n_{1}} \ema,
	\end{eqnarray}
	where $M_{1,1},\cdots, M_{s_1,1}$ are linearly independent. From (\ref{M2}), it is obvious that $M$ satisfies (\ref{M'}) if $s_1=m_1$ or $M_{i,j}$ is linearly dependent with $M_{1,1},\cdots, M_{s_1,1}$ for any $s_1+1\le i\le m_1$ and $2\le j\le n_1$. 
	
	{\textbf{Step}} \hypertarget{step2}{{\textbf{2}}}
	Suppose $s_1<m_1$ and there exists a block $M_{i,j}$ that is linearly independent with $M_{1,1},\cdots, M_{s_1,1}$ for $s_1+1\le i\le m_1$ and $2\le j\le n_1$. Using Lemma \ref{lemma11}, we multiply an appropriate constant to the $j$-th block-column of $M_1$ and add to the first block-column, such that the new first block-column has $s_1+a$ linearly independent blocks, $1\le a\le m_1-s_1$. Next, by repeating {\textbf{Step}} \hyperlink{step1}{1}, we have 
	\begin{eqnarray}
	\label{M3}
	M_1\sim M_2=\bma M_{1,1}&M_{1,2}&M_{1,3}&\cdots &M_{1,n_{1}}\\ M_{2,1}&M_{2,2}&M_{2,3}&\cdots &M_{2,n_{1}}\\
	\vdots &\vdots &\vdots &\ddots &\vdots\\M_{s_{1}+a,1}&M_{s_{1}+a,2}&M_{s_{1}+a,3}&\cdots&M_{s_{1}+a,n_{1}}\\0&M_{s_{1}+a+1,2}&M_{s_{1}+a+1,3}&\cdots&M_{s_{1}+a+1,n_1}\\\vdots&\vdots&\vdots&\ddots&\vdots\\
	0&M_{m_{1},2}&M_{m_{1},3}&\cdots &M_{m_{1},n_{1}} \ema,
	\end{eqnarray}
	where $M_{1,1},\cdots, M_{s_1+a,1}$ are linearly independent.
	
	{\textbf{Step}} \hypertarget{step3}{{\textbf{3}}}
	For $M_2$, if $s_1+a<m_1$ and there exists a block $M_{i,j}$ that is linearly independent with $M_{1,1},\cdots, M_{s_1+a,1}$ for $s_1+a+1\le i\le m_1$ and $2\le j\le n_1$, then  repeat {\textbf{Steps}} \hyperlink{step1}{2} and  \hyperlink{step1}{1}. Eventually we obtain that $M\sim M'$ in (\ref{M'}). \\
	
	Recall that $M\in \bbM_{m_1,n_1}\ox\bbM_{m_2,n_2}$ in (\ref{M}) has Schmidt rank $K$, and hence $M'$ has Schmidt rank $K$. If $k_1=K$ in $M'$ in (\ref{N}), then $M'\sim N$ with $p=1$ in terms of (\ref{M'}). Hence we suppose
	\begin{eqnarray}
	\label{k1<K}
	k_1<K
	\end{eqnarray}
	in $M'$. Then our next aim is to obtain that 
	\begin{eqnarray}
	\label{M''}
	M'\sim M''=\bma 
	\begin{array}{c|c|ccc}
	M_{1,1}&M_{1,2}&M_{1,3}&\cdots &M_{1,n_{1}}\\ M_{2,1}&M_{2,2}&M_{2,3}&\cdots &M_{2,n_{1}}\\
	\vdots &\vdots &\vdots &\ddots &\vdots\\M_{{k_2},1}&M_{{k_2},2}&M_{{k_2},3}&\cdots&M_{{k_2},n_1}\\
	\cline{2-5}
	M_{{k_2}+1,1}& &\multicolumn{3}{c}{\multirow{3}{*}{$A_2$}}
	\\
	\vdots&  & & & \\M_{k_{1},1}& & & & \\
	\cline{1-1}
	\cline{3-5}
	0&\multicolumn{4}{c}{\multirow{3}{*}{$A_1$}}\\
	\vdots \\
	0
	\end{array}
	\ema,
	\end{eqnarray}
	where $1\le k_2\le k_1\le m_1$, and $M_{1,1},\cdots,M_{k_1,1}, M_{1,2}, \cdots, M_{k_2,2}$ are linearly independent. At the same time, (\ref{A1span}) still holds and 
	\begin{eqnarray}
	\label{A2span}
	A_2\tilde{\in}\lin \{M_{1,1},\cdots,M_{k_1,1}, M_{1,2}, \cdots, M_{k_2,2}\}.
	\end{eqnarray}
	We next apply the following four steps, namely \textbf{Steps} \hyperlink{step4}{4}-\hyperlink{step7}{7}, to achieve this aim.

	{\textbf{Step}} \hypertarget{step4}{{\textbf{4}}}
	Consider the matrix $M'$ in (\ref{M'}). Because of (\ref{A1span}) and (\ref{k1<K}), we obtain that there exists a block $M_{i,j}$ that is linearly independent with $M_{1,1},\cdots, M_{k_1,1}$ for $1\le i\le k_1$ and $2\le j\le n_1$.  By block-row and block-column switches on $M'$ in (\ref{M'}), $M_{i,j}$ becomes $M_{1,2}$, and $M_{1,1},\cdots,M_{k_1,1}, M_{1,2}$ are linearly independent.
	
	{\textbf{Step}} \hypertarget{step5}{{\textbf{5}}}
	For the new $M'$ in (\ref{M'}), if
	there exists a block $M_{i,2}$ that is linearly independent with $M_{1,1},\cdots, M_{k_1,1}, M_{1,2}$ for $2\le i\le k_1$, then by block-row switch, this block becomes $M_{2,2}$. 
	We repeat from the beginning of {\textbf{Step}} \hyperlink{step1}{5}, until we obtain that 
	\begin{eqnarray}
	\label{M4}
	M'\sim M_3=\bma 
	\begin{array}{c| cccc}
	M_{1,1}&M_{1,2}&M_{1,3}&\cdots &M_{1,n_{1}}\\ M_{2,1}&M_{2,2}&M_{2,3}&\cdots &M_{2,n_{1}}\\
	\vdots &\vdots &\vdots &\ddots &\vdots\\M_{{s_2},1}&M_{{s_2},2}&M_{{s_2},3}&\cdots&M_{{s_2},n_1}\\M_{{s_2}+1,1}&M_{{s_2}+1,2}&M_{{s_2}+1,3}&\cdots&M_{{s_2}+1,n_1}
	\\\vdots &\vdots &\vdots &\ddots &\vdots\\M_{k_{1},1}&M_{k_{1},2}&M_{k_{1},3}&\cdots&M_{k_{1},n_{1}}\\
	\hline
	0&
	\multicolumn{4}{|c}{\multirow{3}{*}{$A_1$}}\\
	\vdots& \\
	0
	\end{array}
	\ema,
	\end{eqnarray}
	where $1\le s_2\le k_1$, and $M_{1,1},\cdots,M_{k_1,1}, M_{1,2}, \cdots, M_{s_2,2}$ are linearly independent. At the same time, 
	\begin{eqnarray}
	\label{Mi2}
	M_{i,2}\in \lin \{M_{1,1},\cdots,M_{k_1,1}, M_{1,2}, \cdots, M_{s_2,2}\}
	\end{eqnarray}
	holds for any $s_2+1\le i\le k_1$. 
	
	{\textbf{Step}} \hypertarget{step6}{{\textbf{6}}}
	By (\ref{Mi2}), we apply block-row operations on $M_3$ in (\ref{M4}) such that 
	\begin{eqnarray}
	\label{i2}
	M_{i,2}\in \lin \{M_{1,1},\cdots,M_{k_1,1}\}
	\end{eqnarray}
	holds for any $s_2+1\le i\le k_1$.
	
	{\textbf{Step}} \hypertarget{step7}{{\textbf{7}}} 
	From (\ref{A1span}) and (\ref{i2}), one can obtain that $M_3\sim M''$ in (\ref{M''})  if $s_2=k_1$ or $M_{i,j}$ is linearly dependent with $M_{1,1},\cdots,M_{k_1,1}, M_{1,2}, \cdots, M_{s_2,2}$ for any $s_2+1\le i\le k_1$ and $3\le j\le n_1$.
	
	Suppose $s_2<k_1$ and there exists a block $M_{ij}$ ($s_2+1\le i\le k_1$, $3\le j\le n_1$) that is linearly independent with  $M_{1,1},\cdots,M_{k_1,1}$, $M_{1,2}, \cdots, M_{s_2,2}$ in (\ref{M4}).  Using Lemma \ref{lemma11}, we multiply an appropriate constant to the $j$-th block-column of $M_3$ and add to the second block-column, such that the new second block-column has $s_2+b$ linearly independent blocks, $1\le b\le k_1-s_2$, and they are linearly independent with $M_{1,1},\cdots,M_{k_1,1}$. For the new $M_3$, we next repeat {\textbf{Steps}} \hyperlink{step5}{5}, \hyperlink{step6}{6} and \hyperlink{step7}{7}. Finally we obtain that $M'\sim M''$ in (\ref{M''}), and $M''$ has Schmidt rank $K$.\\
	
	If $k_1+k_2=K$, then  $M''\sim N$ with $p=2$ in (\ref{N}). 
	On the other hand if $k_1+k_2<K$, then we consider the third block-column of $M''$ using the same way from { \textbf{Step}} \hyperlink{step4}{4} to \hyperlink{step7}{7}.  Continue this process until we obtain that $M''\sim N$.
	
	By achieving (\ref{M'}) and (\ref{M''}), we have shown that $M\sim M'\sim M''\sim N$ in (\ref{N}). Hence we have finished the proof.
\end{proof}

\subsection{Equivalent form of $M$}
\label{subsec: part2}
In subsection \ref{subsec: part1}, we have shown that any bipartite matrix $M$ has a canonical form $N\in \bbM_{canonical}$ in (\ref{N}) up to local equivalence. In this subsection, we continue to apply three steps on $N$ , namely {\textbf{Steps}} \hyperlink{step8}{8}-\hyperlink{step10}{10}, to obtain another form of $M$ in (\ref{Np}) up to local equivalence.

{\textbf{Step}} \hypertarget{step8}{{\textbf{8}}}
Denote
$*_i$ as a matrix that contains exactly $i$ columns, $\lambda_i$ as a matrix that contains exactly $i$ column vectors and they are linearly independent, $0_i$ as a matrix that contains exactly $i$ columns and they are all zero column vectors.	
Denote $*_i^T$ as a matrix that contains exactly $i$ rows, $0_i^T$ as a matrix that contains exactly $i$ zero rows. Denote $0_{i}^{\G_B}$ as the partial transpose of system $B$ of a block matrix $0_i$.

Recall that $N\in \bbM_{m_1,n_1}\ox\bbM_{m_2,n_2}$ in (\ref{N}). Consider the $k_1$ linearly independent blocks in the first block-column of $N$, and they form a matrix of $n_2$ columns. If the matrix has $r_1$ linearly independent column vectors, then 
\begin{eqnarray}
r_1\le \min\{n_2, r(N)\}.
\end{eqnarray}
Next, we can find an order-$n_2$ invertible matrix $R_1$ such that
\begin{eqnarray}
\label{fc1}
\bma 
N_{1,1}\\ 
\vdots\\
N_{k_{1},1}
\ema \cdot R_1=
\bma 
\begin{array}{c}
\multirow{3}{*}{$\lambda_{r_1}$ $0_{n_2-r_1}$}\\   \\
\\

\end{array}
\ema,
\end{eqnarray} 
where the leftmost $r_1$ column vectors are linearly independent, and the rightmost $n_2-r_1$ column vectors are zero vectors. 
Denote $N_{i,j}^{(A_k)}$ as a block $N_{i,j}$ in $A_k$, where $1\le k\le p$ in (\ref{N}). Recall that $A_1\tilde{\in} \lin \{N_{1,1},\cdots,N_{k_1,1}\}$ from (\ref{Aspan}).  Using Lemma \ref{lemma6} and (\ref{fc1}), we have 
\begin{eqnarray}
\label{fc2}
N_{i,j}^{(A_1)} \cdot R_1=
\bma 
\begin{array}{c}
\multirow{2}{*}{$*_{r_1}$ $0_{n_2-r_1}$}\\   \\

\end{array}
\ema,
\end{eqnarray} 
where the rightmost $n_2-r_1$ columns are zero column vectors. 
Let 
\begin{eqnarray}  
N_1=N\cdot (I_{n_1}\otimes R_1),
\end{eqnarray}   
From (\ref{N}), (\ref{fc1}) and (\ref{fc2}) we have 
\begin{eqnarray}
\label{N1}
N\sim N_1=\bma 
\begin{array}{c|c|ccc|cc}
\multirow{10}{*}{$\lambda_{r_1}$ $0_{n_2-r_1}$}&N_{1,2}&\multicolumn{1}{c|}{N_{1,3}}&\multicolumn{1}{c|}{\cdots}&N_{1,p}&\multicolumn{2}{c}{\multirow{5}{*}{$A_p$}}\\
&\vdots &\multicolumn{1}{c|}{\vdots} &\multicolumn{1}{c|}{\ddots}&\vdots \\ &N_{{k_p},2}&\multicolumn{1}{c|}{N_{{k_p},3}}&\multicolumn{1}{c|}{\cdots}&N_{{k_p},p} \\
\cline{5-5}
&\vdots &\multicolumn{1}{c|}{\vdots} &\multicolumn{1}{c|}{\vdots} & \multicolumn{1}{c|}{\multirow{2}{*}{$A_{p-1}$}}\\
\cline{4-4}
\cline{6-7}

&\vdots&\multicolumn{1}{c|}{\vdots}& \multicolumn{1}{c|}{}&\multicolumn{1}{c}{} \\
\cline{5-7}
&N_{{k_3},2} &\multicolumn{1}{c|}{N_{{k_3},3}}& \multicolumn{4}{c}{\multirow{3}{*}{$\cdots$}}\\
\cline{3-3}
&\vdots&\multicolumn{1}{c|}{}\\
&N_{{k_2},2}&\multicolumn{1}{c|}{} & \\
\cline{2-2}
\cline{4-7}
& \multirow{2}{*}{$*_{r_1}$ $0_{n_2-r_1}$} &\multicolumn{5}{c}{\multirow{2}{*}{$A_2$}}\\ &  \\
\cline{1-1}
\cline{3-7}
\multirow{3}{*}{$0_{n_2}$}&\multicolumn{1}{c}{\multirow{3}{*}{$*_{r_1}$ $0_{n_2-r_1}$}}&\multicolumn{1}{c}{\multirow{3}{*}{$*_{r_1}$ $0_{n_2-r_1}$}}&\multirow{3}{*}{$\cdots$}&\multicolumn{1}{c}{\multirow{3}{*}{$*_{r_1}$ $0_{n_2-r_1}$}}&\multicolumn{2}{c}{\multirow{3}{*}{$\cdots$}}\\
& \multicolumn{1}{c}{}\\
&\multicolumn{1}{c}{}
\end{array}
\ema,
\end{eqnarray}
where (\ref{sum}), (\ref{K}) and (\ref{Aspan})  still hold by Lemma \ref{lemmaadd}. Note that $N_{i,j}\cdot R_1$ has been denoted by $N_{i,j}$ still, when there is no confusion. The same denotion applies also to $A_2,\cdots, A_p$, as well as the end of {\textbf{Steps}} \hyperlink{step9}{9}  and \hyperlink{step10}{10}.

{\textbf{Step}} \hypertarget{step9}{{\textbf{9}}}
Consider the $k_2$ linearly independent blocks in the second block-column of $N_1$ in (\ref{N1}).  We write 
\begin{eqnarray}
\label{fc3}
\bma 
N_{1,2}\\
\vdots\\
N_{k_{2},2}
\ema=
\bma 
\begin{array}{c}
\multirow{3}{*}{$*_{r_1}$ $*_{n_2-r_1}$}\\   \\
\\
\end{array}
\ema,
\end{eqnarray} 
and next consider the rightmost $n_2-r_1$ columns in (\ref{fc3}). Suppose they form a matrix of $r_2$ linearly independent column vectors, then 
\begin{eqnarray}
r_2\le \min\{n_2-r_1, r(N_1)\}.
\end{eqnarray}
Thus we can find an order-($n_2-r_1$) invertible matrix $R_2$ such that
\begin{eqnarray}
\label{fc4}
\bma 
N_{1,2}\\ 
\vdots\\
N_{k_{2},2}
\ema \cdot 
\bma
I_{r_1}&0\\
0&R_2
\ema
=
\bma 
\begin{array}{c}
\multirow{3}{*}{$*_{r_1}$ $\lambda_{r_2}$ $0_{n_2-r_1-r_2}$}\\   \\
\\
\end{array}
\ema,
\end{eqnarray}  
where the middle $r_2$ column vectors are linearly independent and the rightmost $n_2-r_1-r_2$ column vectors are zero vectors. Further, for $N_1$ in (\ref{N1}), we have 
\begin{eqnarray}
\label{fc5}
N_{i,j}^{(A_1)} \cdot 
\bma
I_{r_1}&0\\
0&R_2
\ema=
\bma
\begin{array}{c}
\multirow{2}{*}{$*_{r_1}$ $0_{n_2-r_1}$} \\ \\ \end{array}
\ema
\cdot 
\bma
I_{r_1}&0\\
0&R_2
\ema=N_{i,j}^{(A_1)}.
\end{eqnarray}
Recall that $A_2\tilde{\in} \lin \{N_{1,1},\cdots,N_{k_1,1},N_{1,2},\cdots,N_{k_2,2}\}$ from (\ref{Aspan}). Using Lemma \ref{lemma6} and (\ref{fc4}), we have 
\begin{eqnarray}
\label{fc6}
N_{i,j}^{(A_2)} \cdot 
\bma
I_{r_1}&0\\
0&R_2
\ema=
\bma 
\begin{array}{c}
\multirow{2}{*}{$*_{r_1+r_2}$ $0_{n_2-r_1-r_2}$}\\  \\

\end{array}
\ema.
\end{eqnarray} 
Let \begin{eqnarray}
N_2=N_1\cdot (I_{n_1}\otimes \bma
I_{r_1}&0\\
0&R_2
\ema).
\end{eqnarray} 
From (\ref{N1}), (\ref{fc4}), (\ref{fc5}) and (\ref{fc6}), we have 
\begin{eqnarray}
\label{N2}
\notag &&N_1\sim N_2=
\\ \notag
&&\bma 
\begin{array}{c|c|c|ccc|cc}
\multirow{12}{*}{$\lambda_{r_1}$ $0_{n_2-r_1}$}&\multirow{10}{*}{$*_{r_1}$ $\lambda_{r_2}$ $0_{n_2-r_1-r_2}$}&N_{1,3}&\multicolumn{1}{c|}{N_{1,4}}&\multicolumn{1}{c|}{\cdots}&N_{1,p}&\multicolumn{2}{c}{\multirow{6}{*}{$A_p$}}\\
&  &\vdots&\multicolumn{1}{c|}{\vdots} &\multicolumn{1}{c|}{\ddots}&\vdots \\& &N_{{k_p},3}&\multicolumn{1}{c|}{N_{{k_p},4}}&\multicolumn{1}{c|}{\cdots}&N_{{k_p},p} \\
\cline{6-6}
&  & \vdots&\multicolumn{1}{c|}{\vdots} & \multicolumn{1}{c|}{\vdots}&\multirow{2}{*}{$A_{p-1}$} & &\\
\cline{5-5}
\cline{7-8}
&  &\vdots &\multicolumn{1}{c|}{\vdots} &\multicolumn{1}{c|}{}& \multicolumn{1}{c}{}&\\
\cline{6-8}
&  &N_{k_4,3} &\multicolumn{1}{c|}{N_{k_4,4}} &\multicolumn{4}{c}{\multirow{3}{*}{$\cdots$}} \\
\cline{4-4}
&  &\vdots &\multicolumn{1}{c|}{} &\multicolumn{1}{c}{}& \multicolumn{1}{c}{}&\\
&  & N_{k_3,3}& \multicolumn{1}{c|}{}&\multicolumn{4}{c}{ }\\
\cline{3-3}
\cline{5-8}
& &\multirow{2}{*}{$*_{r_1+r_2}$ $0_{n_2-r_1-r_2}$} &\multicolumn{5}{c}{\multirow{2}{*}{$A_3$}}\\
& &\multicolumn{1}{c|}{} & \multicolumn{5}{c}{}\\
\cline{2-2}
\cline{4-8}
& \multirow{2}{*}{$*_{r_1}$ $0_{n_2-r_1}$}&\multicolumn{1}{c}{\multirow{2}{*}{$*_{r_1+r_2}$ $0_{n_2-r_1-r_2}$}}&\multirow{2}{*}{$\cdots$}&\multicolumn{1}{c}{\multirow{2}{*}{$\cdots$}}&\multicolumn{1}{c}{\multirow{2}{*}{$*_{r_1+r_2}$ $0_{n_2-r_1-r_2}$}} &\multicolumn{2}{c}{\multirow{2}{*}{$\cdots$}} \\&  &\multicolumn{1}{c}{} & &\multicolumn{1}{c}{} &\multicolumn{1}{c}{} &\\
\cline{1-1}
\cline{3-8}
\multirow{3}{*}{$0_{n_2}$}&\multicolumn{1}{c}{\multirow{3}{*}{$*_{r_1}$ $0_{n_2-r_1}$}}&\multicolumn{1}{c}{\multirow{3}{*}{$*_{r_1}$ $0_{n_2-r_1}$}}&\multirow{3}{*}{$\cdots$}&\multicolumn{1}{c}{\multirow{3}{*}{$\cdots$}} &\multicolumn{1}{c}{\multirow{3}{*}{$*_{r_1}$ $0_{n_2-r_1}$}}& \multicolumn{2}{c}{\multirow{3}{*}{$\cdots$}}  \\   &\multicolumn{1}{c}{} &\multicolumn{1}{c}{} & &\multicolumn{1}{c}{}& \multicolumn{1}{c}{} \\
&\multicolumn{1}{c}{} & \multicolumn{1}{c}{}& &\multicolumn{1}{c}{}&\multicolumn{1}{c}{} \\

\end{array}
\ema,\\
\end{eqnarray}
where (\ref{sum}), (\ref{K}) and (\ref{Aspan}) still hold by Lemma \ref{lemmaadd}. Note that if $r_1=n_2$, then we have $N_1=N_2$, i.e., the rightmost $n_2-r_1$ columns of each block in $N_2$ disappear.

{\textbf{Step}} \hypertarget{step10}{{\textbf{10}}}
Consider the $k_3$ linearly independent blocks in the third block-column of $N_2$ using the same way in {\textbf{Steps}} \hyperlink{step8}{8}  and \hyperlink{step9}{9}. Suppose the rightmost $n_2-r_1-r_2$ columns of the block matrix have $r_3$ linearly independent column vectors. 
Then
\begin{eqnarray}
r_3\le \min\{n_2-r_1-r_2, r(N_2)\}.
\end{eqnarray}
We can find an order-($n_2-r_1-r_2$) invertible matrix $R_3$ such that
\begin{eqnarray}
\label{fc7}
\bma 
N_{1,3}\\ 
\vdots\\
N_{k_{3},3}
\ema \cdot 
\bma
I_{r_1+r_2}&0\\
0&R_3
\ema
=
\bma 
\begin{array}{c}
\multirow{3}{*}{$*_{r_1+r_2}$ $\lambda_{r_3}$ $0_{n_2-r_1-r_2-r_3}$}\\   \\
\\
\end{array}
\ema.
\end{eqnarray} 
Let 
\begin{eqnarray}
N_3=N_2\cdot (I_{n_1}\otimes \bma
I_{r_1+r_2}&0\\
0&R_3
\ema).
\end{eqnarray}
Hence $N_2\sim N_3$, and we next consider the fourth block-column of $N_3$. 
Continuing this process and using (\ref{N1}), (\ref{N2}), finally we obtain that 
{\large
	\begin{eqnarray}
	\label{Np}
	\notag &&N\sim N_1\sim \cdots \sim N_p
	=\\
	\notag 
	&&\bma 
	\begin{array}{c|c|c|c|c|c|c}
	\multirow{10}{*}{$\lambda_{r_1}$ $0_{n_2-r_1}$}&\multirow{8}{*}{$*_{r_1}$ $\lambda_{r_2}$ $0_{n_2-r_1-r_2}$}&\multirow{6}{*}{$*_{r_1+r_2}$ $\lambda_{r_3}$ $0_{n_2-r_1-r_2-r_3}$}&\cdots &\multirow{2}{*}{$*_ {r_1+\cdots+r_{p-1}} \lambda_{r_p} 0_{n_2-r_1-\cdots-r_p}$}&\multicolumn{2}{c}{\multirow{4}{*}{\textbf{$A_p$}}}\\  & & &\ddots & & \multicolumn{2}{c}{}\\
	\cline{5-5}
	&\multirow{5}{*}{}&\multirow{3}{*}{}&\cdots&\multirow{2}{*}{$*_{r_1+\cdots+r_{p-1}} 0_{n_2-r_1-\cdots r_{p-1}}$}&\multicolumn{2}{c}{\multirow{2}{*}{}}\\
	&\multirow{5}{*}{}&\multirow{3}{*}{}& \vdots& & \multicolumn{2}{c}{} \\
	\cline{4-4}
	\cline{6-7}
	& & &  &\multicolumn{1}{c}{\multirow{2}{*}{$*_{r_1+\cdots+r_{p-1}} 0_{n_2-r_1-\cdots r_{p-1}}$}} &\multicolumn{2}{c}{\multirow{2}{*}{$\cdots$}} \\
	
	& & &\multicolumn{1}{c|}{} &\multicolumn{1}{c}{ }&\multicolumn{2}{c}{} \\
	\cline{3-3}
	\cline{5-7}
	&\multirow{5}{*}{}&\multirow{2}{*}{$*_{r_1+r_2}$ $0_{n_2-r_1-r_2}$}&\multicolumn{4}{c}{\multirow{2}{*}{$\cdots$}} \\
	&\multirow{5}{*}{}& &\multicolumn{1}{c}{} &\multicolumn{1}{c}{} & \multicolumn{2}{c}{}\\
	\cline{2-2}
	\cline{4-7}
	& \multirow{2}{*}{$*_{r_1}$ $0_{n_2-r_1}$}&\multicolumn{1}{c}{\multirow{2}{*}{$*_{r_1+r_2}$ $0_{n_2-r_1-r_2}$}}&\multicolumn{1}{c}{\multirow{2}{*}{$\cdots$}}&\multicolumn{1}{c}{\multirow{2}{*}{$*_{r_1+r_2}$ $0_{n_2-r_1-r_2}$}}&\multicolumn{2}{c}{\multirow{2}{*}{$\cdots$}}\\
	& &\multicolumn{1}{c}{} &\multicolumn{1}{c}{} &\multicolumn{1}{c}{} & \multicolumn{2}{c}{}\\
	\cline{1-1}
	\cline{3-7}
	\multirow{3}{*}{$0_{n_2}$}&\multicolumn{1}{c}{\multirow{3}{*}{$*_{r_1}$ $0_{n_2-r_1}$}}&\multicolumn{1}{c}{\multirow{3}{*}{$*_{r_1}$ $0_{n_2-r_1}$}} &\multicolumn{1}{c}{\multirow{3}{*}{$\cdots$}}& \multicolumn{1}{c}{\multirow{3}{*}{$*_{r_1}$ $0_{n_2-r_1}$}}& \multicolumn{2}{c}{\multirow{3}{*}{$\cdots$}}\\
	&\multicolumn{1}{c}{} &\multicolumn{1}{c}{} &\multicolumn{1}{c}{} &\multicolumn{1}{c}{} &\multicolumn{2}{c}{}\\
	&\multicolumn{1}{c}{} & \multicolumn{1}{c}{}& \multicolumn{1}{c}{}&\multicolumn{1}{c}{} &\multicolumn{2}{c}{} 
	\end{array}
	\ema,\\
	\end{eqnarray}}
where $*_i, \lambda_i, 0_i$ are defined in {\textbf{Step}} \hyperlink{step8}{8} and $r_1+r_2+\cdots+r_p\le n_2$. At the same time, (\ref{sum}), (\ref{K}) and (\ref{Aspan}) still hold by Lemma \ref{lemmaadd}. At the same time, the rightmost ($n_2-r_1-\cdots-r_p$) column vectors of $N_{i,j}^{(A_p)}$ in (\ref{Np}) are zero vectors by (\ref{Aspan}) and Lemma \ref{lemma6}. One can also verify that if $r_1+\cdots+r_i=n_2$ for $1\le i<p$, then $N\sim N_{i}=N_{i+1}=\cdots=N_{p}$.

We have shown that $N\sim N_p$ in (\ref{N}) and (\ref{Np}), hence we have presented an equivalent form of $M$ in (\ref{M}).
\subsection{Proof of Conjecture \ref{cj:1}}
\label{subsec: part3}
In this subsection, we prove the conjectue by induction.
\begin{theorem}
	\label{thm:main}
	Conjecture \ref{cj:1} holds. Equivalently, the inequality \eqref{eq:main-inequality} holds for any tripartite mixed state $\r_{ABC}$.	
\end{theorem}
\begin{proof}
	 First, one can show that Conjecture \ref{cj:1} holds for any $M$ of Schmidt rank one. Next, suppose Conjecture \ref{cj:1} holds for any matrix of Schmidt rank at most $K-1$, with $K\ge 2$. We will prove that Conjecture \ref{cj:1} holds for any $M$ of Schmidt rank $K$.

	In subsections \ref{subsec: part1} and \ref{subsec: part2}, we have shown that proving $M$ satisfies Conjecture \ref{cj:1} is equivalent to proving $N_p\sim N$ in (\ref{Np}) satisfies Conjecture \ref{cj:1}. 
	Further, it has been proved by Lemma \ref{lemma10} that $N$ in (\ref{N}) satisfies Conjecture \ref{cj:1} if $k_1=K$ in (\ref{K}). So from (\ref{sum}) and (\ref{K}), we assume that
	\begin{eqnarray}
	\label{k1lessK}
	1\le k_s<K
	\end{eqnarray}
	holds for any $1\le s\le p$ in $N$ in (\ref{N}) and $N_p$ in (\ref{Np}).
	
	We next decompose $N_p$ in (\ref{Np}) into the sum of $p$ block matrices.
	Firstly, let  $N_{q_1}=[N_{i,j}^{(q_1)}]\in \bbM_{m_1,n_1}\ox\bbM_{m_2,n_2}$, where the first $r_1$ columns of $N_{i,j}^{(q_1)}$ in $N_{q_1}$ are exactly the first $r_1$ columns of $N_{i,j}$ in $N_p$ in (\ref{Np}), the remaining $n_2-r_1$ column vectors of $N_{i,j}^{(q_1)}$ in $N_{q_1}$ are zero vectors. 
	From (\ref{Np}), we have
	\begin{eqnarray}
	\notag
	N_{q_1}
	=&&\bma 
	\begin{array}{c|c|c|c}
	\multirow{3}{*}{$\lambda_{r_1}$ $0_{n_2-r_1}$}&\multirow{3}{*}{$*_{r_1}$ $0_{n_2-r_1}$} &\multirow{3}{*}{$\cdots$} &\multirow{3}{*}{$*_{r_1}$ $0_{n_2-r_1}$} \\
	& & & \\
	& & & \\
	\hline
	\multirow{2}{*}{$0_{n_2}$}&\multirow{2}{*}{$*_{r_1}$ $0_{n_2-r_1}$}&\multirow{2}{*}{$\cdots$} &\multirow{2}{*}{$*_{r_1}$ $0_{n_2-r_1}$}\\
	& & & 
	\end{array}
	\ema\\
	\label{Np11}
	:=&&\bma
	\begin{array}{c|ccc}
	N_{1,1}^{(q_1)}&N_{1,2}^{(q_1)}&\cdots&N_{1,n_1}^{(q_1)}\\
	\vdots&\vdots&\ddots&\vdots\\
	N_{k_1,1}^{(q_1)}&N_{k_1,2}^{(q_1)}&\cdots&N_{k_1,n_1}^{(q_1)}\\
	\hline
	\multirow{2}{*}{{\large$0_{n_2}$}}&\multicolumn{3}{c}{\multirow{2}{*}{{\huge $\omega$}$_{1}$}}\\
	& &
	\end{array}
	\ema,
	\end{eqnarray}
	where {\Large $\omega_{1}$} is an $(m_1-k_1)\times (n_1-1)$ rectangular block matrix, and $N_{i,1}^{(q_1)}=N_{i,1}$ in terms of (\ref{N}) and (\ref{Np}), $i=1,\cdots,k_1$. From (\ref{Aspan}), we obtain that 
	\begin{eqnarray}
	\label{fc9}
	\text{\Large $\omega_{1}$} \tilde{\in} \lin \{N_{1,1}^{(q_1)},\cdots,N_{k_1,1}^{(q_1)}\},
	\end{eqnarray}
	holds in (\ref{Np11}). Thus by (\ref{k1lessK}), we have
	\begin{eqnarray}
	\label{sr1}
	Sr(\text{\Large $\omega_{1}$})\le k_1<K,
	\end{eqnarray}
	i.e., the Schmidt rank of {\Large $\omega_{1}$} is less than $K$. At the same time, from (\ref{Np11}), using Lemma \ref{lemma1} and the definition of $\lambda_i$, we have
	\begin{eqnarray}
	\label{less6}
	r(\bma N_{1,1}^{(q_1)} \\ \vdots \\N_{k_1,1}^{(q_1)} \ema)+r(\text{\Large $\omega_{1}$})=
	r_1+r(\text{\Large $\omega_{1}$})\le r(N_{q_1}).
	\end{eqnarray}
	Secondly, let $N_{q_2}=[N_{i,j}^{(q_2)}]\in \bbM_{m_1,n_1}\ox\bbM_{m_2,n_2}$, where the $k$-th column of $N_{i,j}^{(q_2)}$ in $N_{q_2}$ is exactly the $k$-th column of $N_{i,j}$ in $N_p$ in (\ref{Np}), $r_1+1\le k\le r_1+r_2$. At the same time, the remaining $n_2-r_2$ column vectors of $N_{i,j}^{(q_2)}$ in $N_{q_2}$ are zero vectors. From (\ref{Np}), we have
	\begin{eqnarray}
	\notag
	N_{q_2}
	=&&
	\bma
	\begin{array}{c|c|c|c|c}
	\multirow{7}{*}{$0_{n_2}$}&\multirow{3}{*}{$0_{r_1} \lambda_{r_2} 0_{n_2-r_1-r_2}$}&\multirow{3}{*}{$0_{r_1} *_{r_2} 0_{n_2-r_1-r_2}$}&\multirow{3}{*}{$\cdots$}&\multirow{3}{*}{$0_{r_1} *_{r_2} 0_{n_2-r_1-r_2}$}\\
	& & & & \\
	& & & & \\
	\cline{2-5}
	& \multirow{2}{*}{$0_{n_2}$}&\multirow{2}{*}{$0_{r_1} *_{r_2} 0_{n_2-r_1-r_2}$}&\multirow{2}{*}{$\cdots$}&
	\multirow{2}{*}{$0_{r_1} *_{r_2} 0_{n_2-r_1-r_2}$}\\
	& & & &\\
	\cline{2-5}
	&\multirow{2}{*}{$0_{n_2}$}&\multirow{2}{*}{$0_{n_2}$}&\multirow{2}{*}{$\cdots$}&\multirow{2}{*}{$0_{n_2}$}\\
	& & & &
	\end{array}
	\ema\\
	\label{Np21}
	:=&&\bma 
	\begin{array}{c|c|ccc}
	\multirow{7}{*}{$0_{n_2}$}&N_{1,2}^{(q_2)}&N_{1,3}^{(q_2)}&\cdots&N_{1,n_1}^{(q_2)}\\
	&\vdots&\vdots&\ddots&\vdots\\
	&N_{k_2,2}^{(q_2)}&N_{k_2,3}^{(q_2)}&\cdots&N_{k_2,n_1}^{(q_2)}\\
	\cline{2-5}
	&\multirow{2}{*}{{\large$0_{n_2}$}}&\multicolumn{3}{c}{\multirow{2}{*}{{\huge $\omega$}$_{2}$}}\\
	& & & & \\
	\cline{2-5}
	&\multicolumn{4}{c}{\multirow{2}{*}{$0_{(n_1-1)n_2}$}}\\
	& \multicolumn{1}{c}{}&
	
	\end{array}
	\ema,
	\end{eqnarray}
	where  {\Large $\omega_{2}$} is a $(k_1-k_2)\times (n_1-2)$ rectangular block matrix. Recall that ({\ref{Aspan}}) holds in $N_p$ in (\ref{Np}), we have $\text{\Large $\omega_{2}$} \tilde{\in} \lin \{N_{1,1}^{(q_2)},\cdots,N_{k_1,1}^{(q_2)},N_{1,2}^{(q_2)},\cdots,N_{k_2,2}^{(q_2)}\}$ in $N_{q_2}$.  Note that  $N_{1,1}^{(q_2)},\cdots,N_{k_1,1}^{(q_2)}$ are zero matrices,
	therefore
	\begin{eqnarray}
	\label{fc10}
	\text{\Large $\omega_{2}$} \tilde{\in} \lin \{N_{1,2}^{(q_2)},\cdots,N_{k_2,2}^{(q_2)}\},
	\end{eqnarray}
	and by (\ref{k1lessK}), we have
	\begin{eqnarray}
	Sr(\text{\Large $\omega_{2}$})\le k_2<K.
	\end{eqnarray}
	At the same time, using Lemma \ref{lemma1} and (\ref{Np21}), we have
	\begin{eqnarray}
	\label{less7}
	r(\bma N_{1,2}^{(q_2)} \\ \vdots \\N_{k_2,2}^{(q_2)} \ema)+r(\text{\Large $\omega_{2}$})=
	r_2+r(\text{\Large $\omega_{2}$})\le r(N_{q_2}). 
	\end{eqnarray}
	Similarly, we continue to find $N_{q_s}=[N_{i,j}^{(q_s)}] \in \bbM_{m_1,n_1}\ox\bbM_{m_2,n_2}$, where $3\le s\le p$. For each $N_{q_s}$, the $k$-th column of $N_{i,j}^{(q_s)}$ is exactly  the $k$-th column of $N_{i,j}$ in $N_p$ in (\ref{Np}), $r_1+\cdots+r_{s-1}+1\le k\le r_1+\cdots+r_{s}$, and the remaining column vectors of $N_{i,j}^{(q_s)}$  are zero  vectors. Thus by (\ref{Np11}) and (\ref{Np21}), we obtain that
	{\large
		\begin{eqnarray}
		N_{q_s}&&=
		\notag
		\bma
		\begin{array}{c|c|c|c|c}
		\multirow{7}{*}{$0_{(s-1)n_2}$  }&\multirow{3}{*}{$0_{\sum_{i=1}^{s-1}r_i} \lambda_{r_s} 0_{n_2-\sum_{i=1}^{s}r_i}$}&\multirow{3}{*}{$0_{\sum_{i=1}^{s-1}r_i} *_{r_s} 0_{n_2-\sum_{i=1}^{s}r_i}$}&\multirow{3}{*}{$\cdots$}&\multirow{3}{*}{$0_{\sum_{i=1}^{s-1}r_i} *_{r_s} 0_{n_2-\sum_{i=1}^{s}r_i}$}\\
		& & & & \\
		& & & & \\
		\cline{2-5}
		& \multirow{2}{*}{$0_{n_2}$}&\multirow{2}{*}{$0_{\sum_{i=1}^{s-1}r_i} *_{r_s} 0_{n_2-\sum_{i=1}^{s}r_i}$}&\multirow{2}{*}{$\cdots$}&
		\multirow{2}{*}{$0_{\sum_{i=1}^{s-1}r_i} *_{r_s} 0_{n_2-\sum_{i=1}^{s}r_i}$}\\
		& & & &\\
		\cline{2-5}
		&\multirow{2}{*}{$0_{n_2}$}&\multirow{2}{*}{$0_{n_2}$}&\multirow{2}{*}{$\cdots$}&\multirow{2}{*}{$0_{n_2}$}\\
		& & & &
		\end{array}
		\ema
		\end{eqnarray}
	}
	\begin{eqnarray}
	\label{Nqs}
	&&:=\bma 
	\begin{array}{c|c|ccc}
	\multirow{8}{*}{$0_{(s-1)n_2}$}&N_{1,s}^{(q_s)}&N_{1,s+1}^{(q_s)}&\cdots&N_{1,n_1}^{(q_s)}\\
	&\vdots&\vdots&\vdots&\vdots\\
	&N_{k_s,s}^{(q_s)}&N_{k_s,s+1}^{(q_s)}&\cdots&N_{k_s,n_1}^{(q_s)}\\
	\cline{2-5}
	&\multirow{2}{*}{{\large$0_{n_2}$}}&\multicolumn{3}{c}{\multirow{2}{*}{{\huge $\omega$}$_{s}$}}\\
	& & & & \\
	\cline{2-5}
	&\multicolumn{4}{c}{\multirow{3}{*}{$0_{(n_1-s+1)n_2}$}}\\
	& \multicolumn{4}{c}{}\\
	&\multicolumn{4}{c}{}
	\end{array}
	\ema,
	\end{eqnarray}
	where {\Large $\omega_{s}$} is a $(k_{s-1}-k_s)\times (n_1-s)$ rectangular block matrix, and
	\begin{eqnarray}
	\label{fc11}
	&\text{\Large $\omega_{s}$} \tilde{\in} \lin \{N_{1,s}^{(q_s)},\cdots,N_{k_s,s}^{(q_s)}\},
	\end{eqnarray}
	\begin{eqnarray}
	\label{less8}
	r(\bma N_{1,s}^{(q_s)} \\ \vdots \\N_{k_s,s}^{(q_s)} \ema)+r(\text{\Large $\omega_{s}$})=
	r_s+r(\text{\Large $\omega_{s}$})\le r(N_{q_s}) 
	\end{eqnarray}
	hold for any $1\le s\le p$.
	Further, from (\ref{k1lessK}) and (\ref{fc11}), we have 
	\begin{eqnarray}
	\label{KS} 
	Sr({\Large \text{\Large $\omega_{s}$}})\le k_s<K
	\end{eqnarray}
	holds for any $1\le s\le p$. 
	On the other hand, from the construction of $N_{q_s}$, where $1\le s\le p$, one can obtain that
	\begin{eqnarray}
	\label{decompose}
	N_p=N_{q_1}+N_{q_2}+\cdots+N_{q_p},
	\end{eqnarray}
	and by Lemma \ref{lemma1},
	\begin{eqnarray}
	\label{less}
	r(N_{p}) \ge r(N_{q_s})
	\end{eqnarray}
	holds for any $1\le s\le p$.
	
	By (\ref{decompose}), we have decomposed $N_p$ in (\ref{Np}) into the sum of $N_{q_1},\cdots,N_{q_s}$. Note that if $r_1+\cdots+r_i=n_2$ for $1\le i<p$  in $N_p$, then $N_{q_{i+1}},\cdots,N_{q_p}$ disappear. We next consider the partial transpose of system $B$ of each $N_{q_s}$. For $s=1$, from (\ref{Np11}), we have
	\begin{eqnarray}
	\label{Np1T}
	N_{q_1}^{\G_B}
	=&&\bma
	\begin{array}{c|ccc}
	{N_{1,1}^{(q_1)}}^T&{N_{1,2}^{(q_1)}}^T&\cdots&{N_{1,n_1}^{(q_1)}}^T\\
	\vdots&\vdots&\ddots&\vdots\\
	{N_{k_1,1}^{(q_1)}}^T&{N_{k_1,2}^{(q_1)}}^T&\cdots&{N_{k_1,n_1}^{(q_1)}}^T\\
	\hline
	\multirow{2}{*}{{\large$0_{n_2}^{\G_B}$}}&\multicolumn{3}{c}{\multirow{2}{*}{{\huge $\omega$}$_{1}$$^{\G_B}$}}\\
	& &
	\end{array}
	\ema=
	\bma
	\begin{array}{cccc}
	
	*_{r_1}^T&*_{r_1}^T&\multirow{2}{*}{$\cdots$}&*_{r_1}^T\\
	0_{n_2-r_1}^T&0_{n_2-r_1}^T& &0_{n_2-r_1}^T\\
	\hline
	\vdots&\vdots&\ddots&\vdots\\
	\hline
	*_{r_1}^T&*_{r_1}^T&\multirow{2}{*}{$\cdots$}&*_{r_1}^T\\
	0_{n_2-r_1}^T&0_{n_2-r_1}^T& &0_{n_2-r_1}^T\\
	\hline
	\multicolumn{1}{c|}{\multirow{2}{*}{{\large$0_{n_2}^{\G_B}$}}}&\multicolumn{3}{c}{\multirow{2}{*}{{\huge $\omega$}$_{1}$$^{\G_B}$}}\\
	\multicolumn{1}{c|}{}& & & 
	\end{array}
	\ema.
	\end{eqnarray}
	Using Lemma \ref{lemma1} and (\ref{Np1T}), we have
	\begin{eqnarray}
	\label{less3}
	r(N_{q_1}^{\G_B})
	\le k_1\cdot r_1+r( \text{\Large $\omega_{1}$}^{\G_B}).
	\end{eqnarray}
	Similar to $N_{q_1}$, from (\ref{Nqs}), for any $1\le s \le p$, one can show that 
	\begin{eqnarray}
	\label{NpsT}
	\notag
	N_{q_s}^{\G_B}
	&&=\bma 
	\begin{array}{c|c|ccc}
	\multirow{7}{*}{{\large$0_{(s-1)n_2}^{\G_B}$}}&{N_{1,s}^{(q_s)}}^T&{N_{1,s+1}^{(q_s)}}^T&\cdots&{N_{1,n_1}^{(q_s)}}^T\\
	&\vdots&\vdots&\ddots&\vdots\\
	&{N_{k_s,s}^{(q_s)}}^T&{N_{k_s,s+1}^{(q_s)}}^T&\cdots&{N_{k_s,n_1}^{(q_s)}}^T\\
	\cline{2-5}
	&\multirow{2}{*}{$0_{n_2}^{\G_B}$}&\multicolumn{3}{c}{\multirow{2}{*}{{\huge $\omega$}$_{s}$$^{\G_B}$}}\\
	& & & & \\
	\cline{2-5}
	&\multicolumn{4}{c}{\multirow{3}{*}{$0_{(n_1-s_1+1)n_2}^{\G_B}$}}\\
	& \multicolumn{4}{c}{}\\
	&\multicolumn{4}{c}{}
	\end{array}
	\ema \\
	&&=\bma
	\begin{array}{c|cccc}
	\multirow{4}{*}{$0_{n_2}^T$}&0_{\sum_{i=1}^{s-1}r_i}^T&0_{\sum_{i=1}^{s-1}r_i}^T&\multirow{4}{*}{$\cdots$}&0_{\sum_{i=1}^{s-1}r_i}^T\\
	&\multirow{2}{*}{$*_{r_s}^T$}&\multirow{2}{*}{$*_{r_s}^T$}& &\multirow{2}{*}{$*_{r_s}^T$}\\
	& & & & \\
	&0_{n_2-\sum_{i=1}^{s}r_i}^T&0_{n_2-\sum_{i=1}^{s}r_i}^T& &0_{n_2-\sum_{i=1}^{s}r_i}^T\\
	\hline
	\vdots&\vdots&\vdots&\ddots&\vdots\\
	\hline
	\multirow{4}{*}{$0_{n_2}^T$}&0_{\sum_{i=1}^{s-1}r_i}^T&0_{\sum_{i=1}^{s-1}r_i}^T&\multirow{4}{*}{$\cdots$}&0_{\sum_{i=1}^{s-1}r_i}^T\\
	&\multirow{2}{*}{$*_{r_s}^T$}&\multirow{2}{*}{$*_{r_s}^T$}& &\multirow{2}{*}{$*_{r_s}^T$}\\
	& & & & \\
	&0_{n_2-\sum_{i=1}^{s}r_i}^T&0_{n_2-\sum_{i=1}^{s}r_i}^T& &0_{n_2-\sum_{i=1}^{s}r_i}^T\\
	\hline
	\multirow{5}{*}{$0_{(s-1)n_2}^{\G_B}$}&\multicolumn{1}{c|}{\multirow{2}{*}{$0_{n_2}^{\G_B}$}}&\multicolumn{3}{c}{\multirow{2}{*}{{\huge $\omega$}$_{s}$$^{\G_B}$}}\\
	&\multicolumn{1}{c|}{} & & & \\
	\cline{2-5}
	&\multicolumn{4}{c}{\multirow{3}{*}{$0_{(n_1-s_1+1)n_2}^{\G_B}$}} \\
	&\multicolumn{1}{c}{}  \\
	&\multicolumn{1}{c}{}
	\end{array}
	\ema,
	\end{eqnarray}
	and hence
	\begin{eqnarray}
	\label{less13}
	r(N_{q_s}^{\G_B})
	\le k_s\cdot r_s+r( \text{\Large $\omega_{s}$}^{\G_B}).
	\end{eqnarray} 
	From (\ref{KS}), we obtain that the Schmidt rank of {\Large $\omega_{s}$} is at most $K-1$, where {\Large $\omega_{s}$} is from (\ref{Nqs}). Recalling the assumption of induction that Conjecture \ref{cj:1} holds for any matrix of Schmidt rank at most $K-1$, hence for any $1\le s \le p$, we have
	\begin{eqnarray}
	\label{less14}
	r(\text{\Large $\omega_{s}$}^{\G_B})\le Sr(\text{\Large $\omega_{s}$})\cdot r(\text{\Large $\omega_{s}$}).
	\end{eqnarray} 
	Further, from (\ref{less8}), (\ref{KS}), (\ref{less13}) and (\ref{less14}) we obtain that
	\begin{eqnarray}
	\label{less15}
	\notag
	r(N_{q_s}^{\G_B})&&\le k_s\cdot r_s+Sr(\text{\Large $\omega_{s}$})\cdot r(\text{\Large $\omega_{s}$}) \\
	\notag
	&&\le k_s\cdot (r_s+r(\text{\Large $\omega_{s}$}))\\
	&&\le k_s\cdot r(N_{q_s})
	\end{eqnarray}
	holds for any $1\le s \le p$.
	On the other hand, from (\ref{decompose}), we have
	\begin{eqnarray}
	\label{decompose1}
	N_p^{\G_B}=N_{q_1}^{\G_B}+N_{q_2}^{\G_B}+\cdots+N_{q_p}^{\G_B},
	\end{eqnarray}
	and by Lemma \ref{lemma1},
	\begin{eqnarray}
	\label{less1}
	r(N_p^{\G_B})\le r(N_{q_1}^{\G_B})+r(N_{q_2}^{\G_B})+\cdots+r(N_{q_p}^{\G_B}).
	\end{eqnarray}
	From (\ref{K}), (\ref{less}), (\ref{less15}) and (\ref{less1}), we have
	\begin{eqnarray}
	\label{finalless}
	\notag
	r(N_p^{\G_B})&&\le k_1\cdot r(N_{q_1})+k_2\cdot r(N_{q_2})+\cdots+k_p\cdot r(N_{q_p})\\
	\notag
	&&\le k_1\cdot r(N_{p})+k_2\cdot r(N_{p})+\cdots+k_p\cdot r(N_{p})\\
	\notag
	&&=(k_1+k_2+\cdots+k_p)\cdot r(N_{p})\\
	&&=K\cdot r(N_{p}).
	\end{eqnarray}
	This implies Conjecture \ref{cj:1} holds for $N_p$ in (\ref{Np}), and for $M$ in (\ref{M}) up to local equivalence. 
	
	We have shown by induction that Conjecture \ref{cj:1} holds.
	Hence we have finished the proof.
\end{proof}

{\textbf{Example}} To illustrate our proof, we present a boundary case as an example.
Suppose $M$ in (\ref{M}) has $m_1\cdot n_1$ linearly independent blocks. In this case, we have $p=n_1$, $m_1=k_1=k_2=\cdots=k_p$ and $A_1,\cdots,A_p$ disappear in $N$ in (\ref{N}). Further, {\Large $\omega_{s}$} disappears in $N_{q_s}$ in (\ref{Nqs}), $1\le s\le p$. Therefore from (\ref{less8}) and (\ref{NpsT}) we have $r(N_{q_s})^{\G_B} \le  k_s\cdot r_s \le k_s\cdot r(N_{q_s})$ for any $1\le s\le p$.  This implies inequality (\ref{less15}) holds, and further (\ref{finalless}) holds. Hence $M$ satisfies Conjecture \ref{cj:1} up to local equivalence.

\section{Application}
\label{sec:app}

In this section, we apply Theorem \ref{thm:main}. Firstly, we know that there are three inequalities for a tripartite state $\r_{ABC}$ in terms of the 0-entropy $S_0(A)$ of state $\r_A$ \cite{chl14}. Using Theorem \ref{thm:main} the inequalities become
\begin{eqnarray}
\label{eq:0entropy}
&&
S_0(A)+S_0(B)\ge S_0(AB),
\\&&
\label{eq:0entropy-1}	
S_0(AB)+S_0(AC)+S_0(BC)\ge 2S_0(A),
\\&&	
\label{eq:0entropy-2}
S_0(AB)+S_0(AC)\ge S_0(BC).
\end{eqnarray}
Note that we have omitted the inequality 
\begin{eqnarray}
\label{eq:0entropy-3}
S_0(AB)+S_0(AC)\ge S_0(A),	
\end{eqnarray}
obtained in \cite{chl14}, which is a corollary of \eqref{eq:0entropy-1} and \eqref{eq:0entropy-2}. So we have established a complete picture of the four-party linear inequalities in terms of the $0$-entropy.

Next we also point out that, the inequalities \eqref{eq:0entropy}-\eqref{eq:0entropy-2} are necessary conditions by which three bipartite states $\r_{AB}$, $\r_{AC}$ and $\r_{BC}$ are from a tripartite state. For example, consider three bipartite states $\r_{AB}={1\over d^2}I_d \ox I_d$, $\r_{BC}=\r_{AC}=\proj{\ps}$ where
$\ket{\ps}={1\over \sqrt d}\sum^d_{i=1} \ket{ii}$. One can verify that they share the same one-party reduced density operators, $r_{AB}=d^2$ and
$r_{AC}=r_{BC}=1$, so $r_{AB}>r_{AC}r_{BC}$. Thus, the three bipartite states are not from any tripartite state by Theorem \ref{thm:main}. It shows novel understanding to the long-standing marginal problem. 

Third, we extend the inequality in \eqref{eq:0entropy} to the multipartite system as follows.
\begin{eqnarray}
\label{eq:0entropy-ext}
\sum^n_{j=1}
S_0(A_j)	
\ge 
S_0(A_1...A_n).
\end{eqnarray}
We omit the proof as it is similar to that of \eqref{eq:0entropy}.
In the following we extend \eqref{eq:0entropy-1} and
 \eqref{eq:0entropy-2} as follows.
\begin{lemma}
	\label{le:ab.bc.cd>=ad}
	Let $A_{n+1}\equiv A_1$. Given an $n$-partite mixed state of systems $A_1,...,A_n$, we have 
	\begin{eqnarray}
	&&
	\label{eq:0entropy-1-ext}
	\sum^{n}_{j=1}
	S_0(A_jA_{j+1})	
	\ge 2S_0(A_1),
	\\&&
	\label{eq:0entropy-2-ext}
	\sum^{n-1}_{j=1}
	S_0(A_jA_{j+1})	
	\ge S_0(A_1A_n),
	\\&&
	\label{eq:0entropy-3-ext}
	S_0(A_1A_2A_3)+
	S_0(A_1A_2A_4)+
	S_0(A_1A_3A_4)
	\ge 
	S_0(A_2A_3A_4),
	\\ \notag
	\\&&
	\label{eq:0entropy-4-ext}
	S_0(A_1A_2A_3A_4)+S_0(A_1A_2A_3A_5)+S_0(A_1A_2A_4A_5)+S_0(A_1A_3A_4A_5)\ge S_0(A_2A_3A_4A_5), \notag \\
	\\&&
	\label{eq:0entropy-n-ext}
	\sum^{n}_{k=1}S_0(A_1\cdots A_{k-1}A_{k+1}\cdots A_{n}) \geq S_0(A_1). 
	\end{eqnarray} 
\end{lemma}
\begin{proof}
	Firstly we prove \eqref{eq:0entropy-1-ext}. Using Theorem \ref{thm:main} we have $S_0(A_1A_2)+S_0(A_2A_3)\ge S_0(A_1A_3)$. Using the same idea one can show that 
	\begin{eqnarray}
	\notag
	&&
	\sum^{n}_{j=1}
	S_0(A_jA_{j+1})	
	\\ \ge && 
	S_0(A_1A_{n-1})+S_0(A_{n-1}A_n)
	+S_0(A_nA_1)
	\ge 2S_0(A_1),
	\end{eqnarray}
	where the last inequality follows from \eqref{eq:0entropy-1}. We have proven \eqref{eq:0entropy-1-ext}. 
	
	Next we prove \eqref{eq:0entropy-2-ext}. We apply the induction to $n$. If $n=2$ then the assertion holds. Suppose it holds for $n-1$. We have  
	\begin{eqnarray}
	\notag
	\sum^{n-1}_{j=1}
	S_0(A_jA_{j+1})
	=&& 
	\sum^{n-2}_{j=1}
	S_0(A_jA_{j+1})	
	+
	S_0(A_{n-1}A_n)
	\\
	\notag
	\ge &&
	S_0(A_1A_{n-1})	
	+
	S_0(A_{n-1}A_n)
	\\\ge &&
	S_0(A_1A_n),	
	\end{eqnarray} 	
	where the first inequality follows from the induction hypothesis on $n-1$, and the last inequality follows from Theorem \ref{thm:main}. So the assertion holds for $n$. The induction implies that \eqref{eq:0entropy-2-ext} holds.
	
	Further we prove \eqref{eq:0entropy-3-ext}. We have
	\begin{eqnarray}
\notag
	&&
	S_0(A_1A_2A_3)+
	S_0(A_1A_2A_4)+
	S_0(A_1A_3A_4)
	\\
\notag	
	 \ge && 
	S_0(A_1A_2)+
	S_0(A_1A_3A_4)
	\\ \ge && 
	S_0(A_2A_3A_4).
	\end{eqnarray}
	Here the first inequality follows from \eqref{eq:0entropy-3}, and the second inequality follows from Theorem \ref{thm:main}.
	
	Similarly we prove \eqref{eq:0entropy-4-ext}. We can obtain 
	\begin{eqnarray}
	\notag
	&&
	S_0(A_1A_2A_3A_4)+S_0(A_1A_2A_3A_5)+S_0(A_1A_2A_4A_5)+S_0(A_1A_3A_4A_5)
	\\ 
	\notag
	\ge &&
	S_0(A_4A_5)+S_0(A_2A_3)
	\\ \ge &&
	S_0(A_2A_3A_4A_5),
	\end{eqnarray}
	where the first inequality follows from \eqref{eq:0entropy-2} and the second inequality follows from \eqref{eq:0entropy}.
	
	Next we prove \eqref{eq:0entropy-n-ext}. One can show that 
	\begin{eqnarray}
	\notag
	&&
	2\sum^{n}_{k=1}S_0(A_1\cdots A_{k-1}A_{k+1}\cdots A_{n}) 
	\\
	\notag
	 \ge &&
	\sum^{n}_{j=1}
	S_0(A_jA_{j+1})	
	\\ \ge &&
	2S_0(A_1).
	\end{eqnarray}
	Here the first inequality follows from \eqref{eq:0entropy-2}, and the second inequality follows from \eqref{eq:0entropy-1-ext}. So we have $\sum^{n}_{k=1}S_0(A_1\cdots A_{k-1}A_{k+1}\cdots A_{n}) \ge S_0(A_1)$. This completes the proof.
\end{proof}

We have constructed a few novel inequalities every multipartite state satisfy. It shed new light to the marginal problem for multipartite system, as well as the understanding of von Neumann entropy. 

Fourth, we investigate the condition when the inequality in \eqref{eq:main-inequality} is saturated. That is, if $r(\r_{AB}) \cdot r(\r_{AC})\ge r(\r_{BC})$ then what is the condition by which  $r(\r_{AB}) \cdot r(\r_{AC})= r(\r_{BC})$? We partially answer the problem as follows.

\begin{lemma}
\label{le:saturate}
(i) Suppose $\r_{ABC}$ is a tripartite pure state. Then the condition is $r(\r_B)\cdot r(\r_C)=r(\r_A)$. 

(ii) Suppose $\r_{ABC}=\r_A\otimes\r_{BC}$ is a  tripartite mixed state. Then the condition is $r(\r_A)=1$ and $r(\r_{BC})=r(\r_B)\cdot r(\r_C)$. 

(iii) Suppose $\r_{ABC}=\r_{B}\otimes\r_{AC}$ is a  tripartite mixed state. Then the condition is $r(\r_{A}) \cdot r(\r_{AC})= r(\r_{C})$.

(iv) Suppose $\r_{ABC}=\r_{C}\otimes\r_{AB}$ is a  tripartite mixed state. Then the condition is $r(\r_{A}) \cdot r(\r_{AB})= r(\r_{B})$.

(v) Suppose $\r_{ABC}$ is a tripartite PPT state. Then the condition is $ r(\r_{AB})= r(\r_{B})$, $ r(\r_{AC})= r(\r_{C})$ and $r(\r_{B})\cdot r(\r_{C})=r(\r_{BC})$. 
\end{lemma}
\begin{proof}
(i) The assertion follows from the definition of tripartite pure states. 

(ii) We have $r(\r_A)^2\cdot r(\r_B)\cdot r(\r_C)=r(\r_{AB}) \cdot r(\r_{AC})= r(\r_{BC})\le r(\r_B)\cdot r(\r_C)$. So  assertion (i) holds.

(iii) The assertion can be proven straightforwardly.

(iv) The assertion can be proven using the idea of the proof of (ii).

(v) It is known that the rank of a bipartite PPT state is lower bounded by that of any one of its reduced density operators. Hence \eqref{eq:0entropy} implies  
\begin{eqnarray}
r(\r_{AB})
\cdot 
r(\r_{AC})	
\ge 
r(\r_B) \cdot r(\r_C)
\ge 
r(\r_{BC}).
\end{eqnarray}
If $r(\r_{AB}) \cdot r(\r_{AC})= r(\r_{BC})$ then we obtain assertion (v).
\end{proof}

We point out that the states $\r_{ABC}$ satisfying the conditions of this lemma exist, as we show by the following examples. In (i) we assume $\r_{ABC}=\proj{0,0,0}$. Actually the example applies to all of the five cases in Lemma \ref{le:saturate}, and we shall show more non-trivial examples. In (ii) we assume $\r_{ABC}=\proj{0}_A\otimes\proj{\b}_B\otimes\proj{\g}_C$ where $\b$ and $\g$ are arbitrary states. In (iii), we assume that $\r_{AC}$ is a pure state. In (iv) we assume that $\r_{AB}$ is a pure state. In (iv), we assume that $\r_{ABC}=\proj{0}_A\otimes\proj{\b}_B\otimes\proj{\g}_C$ where $\b$ and $\g$ are arbitrary states. 

\section{Conclusions}
\label{sec:con}

We have proven the inequality $r(\r_{AB}) \cdot r(\r_{AC})\ge r(\r_{BC})$ for any tripartite state $\r_{ABC}$ by proving an equivalent conjecture as well as the construction of a novel canonical form of bipartite matrices. So we have a complete picture of the four-party linear inequalities in terms of the $0$-entropy. We also have applied our results to the marginal problem,  extended the inequality to the scenario of multipartite systems, and discussed the condition when the inequality in \eqref{eq:main-inequality} is saturated. 

We believe that the  canonical form in Theorem \ref{canonical} might be applied to more quantum-information problems concerning bipartite systems. Besides some open problems from this paper are as follows.
\begin{enumerate}
\item 
Whether the lower bounds in the inequalities \eqref{eq:0entropy-ext}-\eqref{eq:0entropy-3-ext} are tight enough is unknown.

\item
Can we extend the inequality \eqref{eq:main-inequality} to multipartite system, apart from \eqref{eq:0entropy-ext}-\eqref{eq:0entropy-3-ext}? 

\item
Although we have provided some conditions by which the inequality in \eqref{eq:main-inequality} is saturated in Lemma \ref{le:saturate}, a general condition is still missing.

\end{enumerate}

\section*{Acknowledgments}
	\label{sec:ack}	

Authors were supported by the NNSF of China (Grant No. 11871089), and the Fundamental Research Funds for the Central Universities (Grant Nos. ZG216S2110).

\bibliographystyle{unsrt}

\bibliography{Thefinalanswer} 

\end{document}